\documentclass[submission,copyright,creativecommons,hidelinks]{eptcs}
\providecommand{\event}{ICE 2015} 

\usepackage[colorinlistoftodos,bordercolor=white]{todonotes}

\usepackage{dokter}

\usetikzlibrary{matrix}
\usepackage{mathtools}
\usepackage{stmaryrd} 
\usepackage{verbatim}
\usepackage[export]{adjustbox}
\usepackage{wrapfig}
\tikzset{elliptic state/.style={draw,ellipse}}

\usepackage{amsthm}
\theoremstyle{plain} 
\newtheorem{theorem}{Theorem}
\newtheorem{lemma}{Lemma}

\newtheorem{corollary}{Corollary}
\theoremstyle{definition}
\newtheorem{definition}{Definition}

\newtheorem{example}{Example}
\theoremstyle{remark}

\usepackage[utf8]{inputenc}

\newcommand\xqed[1]{%
  \leavevmode\unskip\penalty9999 \hbox{}\nobreak\hfill
  \quad\hbox{#1}}
\newcommand\tri{\xqed{$\triangle$}}

\DeclareMathOperator{\BIPa}{\operatorname{\sf bip}_1}
\DeclareMathOperator{\Reoa}{\operatorname{\sf reo}_1}
\DeclareMathOperator{\BIPb}{\operatorname{\sf bip}_2}
\DeclareMathOperator{\Reob}{\operatorname{\sf reo}_2}
\DeclareMathOperator{\fa}{\operatorname{\sf f}_1}
\DeclareMathOperator{\ga}{\operatorname{\sf g}_1}
\DeclareMathOperator{\fb}{\operatorname{\sf f}_2}
\DeclareMathOperator{\gb}{\operatorname{\sf g}_2}
\newcommand{\PA}{\mathrm{PA}}
\newcommand{\Arch}{\mathrm{Arch}}
\newcommand{\LTS}{\mathrm{LTS}}
\newcommand{\CA}{\mathrm{CA}^\pm}
\newcommand{\BC}{\mathrm{IM}}

\newcommand{\mdash}[1][]{---#1}

\newcommand{\ie}[1][\ ]{i.e.{#1}}

\newcommand{\defn}[1]{Definition~\ref{defn:#1}}
\newcommand{\fig}[2][]{Figure~\ref{fig:#2}\ensuremath{#1}}
\newcommand{\tab}[1]{Table~\ref{tab:#1}}
\newcommand{\eq}[1]{(\ref{eqn:#1})}

\newcommand{\ex}[1]{Example~\ref{ex:#1}}
\newcommand{\secn}[1]{Section~\ref{sec:#1}}

\newcommand{\lem}[1]{Lemma~\ref{lem:#1}}
\newcommand{\cor}[1]{Corollary~\ref{cor:#1}}
\newcommand{\thm}[1]{Theorem~\ref{thm:#1}}

\newcommand{\bydef}[1]{\ensuremath{\stackrel{\Delta}{#1}}}
\newcommand{\oftype}{\ensuremath{\!:\!}}

\newcommand{\up}{\ensuremath{up}}
\newcommand{\down}{\ensuremath{dn}}

\newcommand{\dtspliti}{\ensuremath{.}}
\newcommand{\dtsplitii}{\ensuremath{:\ }}
\newcommand{\dtsplit}{\ensuremath{\,/\!/\,}}
\newcommand{\connector}[4]
 {\ensuremath{(#1)\dtspliti{}[{#2}\dtsplitii{}{#3}\dtsplit{}{#4}]}}

\newcommand{\domD}[0]{\mathsf{D}}

\renewcommand{\vec}[1]{\mathbf{#1}}

\newcommand{\head}[1] {\ensuremath{\mathit top(#1)}}
\newcommand{\tail}[1] {\ensuremath{\mathit bot(#1)}}
\newcommand{\support}[1] {\ensuremath{\mathit supp(#1)}}

\newcommand{\cB}{\ensuremath{\mathcal{B}}}

\newcommand{\cC}{\ensuremath{\mathcal{C}}}

\title{Relating BIP and Reo}
\author{Kasper Dokter, Sung-Shik Jongmans, Farhad Arbab \\
\institute{Centrum Wiskunde \& Informatica, \\ Amsterdam, Netherlands}
\and
Simon Bliudze \\ 
\institute{\'Ecole Polytechnique F\'ed\'erale de Lausanne, \\ Lausanne, Switzerland}} 

\def\titlerunning{Relating BIP and Reo}
\def\authorrunning{K. Dokter, S.-S. T. Q. Jongmans, F. Arbab \& S. Bliudze}

\begin{document}

\maketitle

\begin{abstract}
Coordination languages simplify design and development of concurrent systems. 
Particularly, exogenous coordination languages, like BIP and Reo, enable system designers to express the interactions among components in a system explicitly. 
In this paper we establish a formal relation between BI(P) (i.e., BIP without the priority layer) and Reo, by defining transformations between their semantic models. 
We show that these transformations preserve all properties expressible in a common semantics. 
This formal relation comprises the basis for a solid comparison and consolidation of the fundamental coordination concepts behind these two languages. 
Moreover, this basis offers translations that enable users of either language to benefit from the toolchains of the other.
\end{abstract}

\raggedbottom

\section{Introduction}
\label{sec:intro}

\paragraph{Context.} 
Over the past decades, architecture description languages (ADL) and coordination languages have emerged as fundamental tools for tackling complexity in the design of correct-by-construction componentised software systems \cite{Garlan}. 
However, no language has yet emerged as a de facto standard, and no consensus exists on how to properly design such languages, either.  
BIP~\cite{bip06,BliSif07-acp-emsoft} and Reo~\cite{Reo} each addresses this complexity and provides a formal semantic framework, which allows reasoning about and proving correctness of coordination as a first-class entity.  

BIP is a language for the construction of concurrent systems by superposing three layers: behaviour, interaction and priorities.
The layered approach of BIP separates concerns between interaction and computation.
This is essential for component-based design of concurrent systems, because it allows global analysis of the coordination layer and reusability of written code.

Reo is a language for compositional specification of coordination protocols, i.e., protocols modeling the synchronization and dataflow among multiple components.
These protocols consist of graph-like structures, called {\em connectors}.
Reo connectors may compose together to form more complex connectors, allowing reusability and compositional construction of coordination protocols.

We provide a more detailed introduction to BIP and Reo in \secn{overview}.

\paragraph{Motivation.} 
Both BIP and Reo advocate the necessity of separating coordination mechanisms from the coordinated components.
In BIP one refers to this separation as the {\em architecture-based} design approach \cite{BBJS14}. 
Reo literature uses the term {\em exogenous coordination} to describe the same fundamental principle \cite{PA01, Reo, Arbab11}. 
Despite this fundamental agreement, the design choices underlying BIP and Reo differ.
For example, BIP uses stateless interactions, while Reo allows stateful connectors.
Establishing a formal relation between BIP and Reo is necessary to discover fundamental principles that drive the design of coordination languages.

Translations exist between numerous other coordination models and BIP and Reo, individually \cite{CRB+08, BMM11, PC08, TSR11}. 
Hence, a formal relationship between BIP and Reo yields insight, albeit indirect, into the relation of each with a wider range of related work.

Furthermore, establishing a formal relationship between BIP and Reo enables encodings that allow each of the two frameworks to benefit from tools and theoretical results obtained for the other. 
These toolchains include tools for editing, code generation, and model checking.
We refer to \cite{biptools} and \cite{reotools,Arbab11} for details.

\paragraph{Contributions.} 
We relate the most important semantic models of BI(P)\footnote{\label{fn:prio}Although BIP's notion of priority is equally applicable to the constraint automata semantics of Reo, Reo provides no syntax to specify such global priority preferences. Reo does have a weaker priority mechanism to specify local preferences by means of context sensitive channel {\tt LossySync}, that prefers locally maximal dataflow.} (i.e., BIP without the priority layer) and Reo. 
For Reo we consider {\em port automata} and {\em constraint automata}, which model Reo connectors at different levels of abstraction \cite{JA12}. 
For BI(P) we consider {\em BIP architectures} \cite{ABBJS14} and {\em BIP interaction models}, i.e., sets of simple interaction expressions \cite{BBJS14}. 

First, we provide a short summary of BIP and Reo in \secn{overview}.
Then, in \secn{PA2Arch}, we define mappings between port automata and BIP architectures, and show that these distribute over composition modulo semantic equivalence.
Hence, it is possible to compute these translations incrementally, in order to speed them up. 
In \secn{CA2BIPconn}, we define mappings between stateless constraint automata and BIP interaction models.
We show that all transformations preserve all properties of observable dataflow, which, for example, enables one to transfer safety properties established for some generated code, or the results of model checking from one model to the other. 
These mappings in the data-sensitive domain do not distribute over composition, but in \secn{conclusion} we briefly discuss a different translation scheme that still allows incremental translation. There, we discuss also the differences and similarities between BI(P) and Reo and other coordination languages, and point out future work.

\paragraph{Related Work.}
Other authors have related and compared both BIP and Reo to other coordination languages. 
Bruni et al. encode BIP models into Petri nets \cite{BMM11}, and Chkouri et al. present a translation of AADL into BIP \cite{CRB+08}. 
Proen\c{c}a and Clarke provide a detailed comparison between Orc and Reo \cite{PC08}.
Arbab et al. provide a translation of Reo connectors into the Tile Model \cite{ABCLM09}.
Krause compared Reo to Petri nets \cite{Krause09}.
Talcott, Sirjani and Ren connect both ARC and PBRD to Reo by providing mappings between their semantic models \cite{TSR11}.

Although an indirect comparison of BIP and Reo through their respective comparisons with other models, e.g., Petri nets, is certainly possible, the direct and formal translations we present in this paper allows direct translation tools between BIP and Reo, that are otherwise difficult, if not impossible, to construct based on such indirect comparisons.

\section{Overview of BIP and Reo}
\label{sec:overview} 

\subsection{BIP}
\label{sec:bip}

A BIP system consist of a superposition of three layers: Behaviour, Interaction, and Priority.  
The behaviour layer encapsulates all computation, consisting of {\em atomic components} processing sequential code.
{\em Ports} form the interface of a component through which it interacts with other components.  
BIP represents these atomic components as \emph{Labeled Transition Systems} (LTS) having transitions labeled with ports and extended with data stored in local variables. 
The second layer defines component coordination by means of
\emph{BIP interaction models}~\cite{BBJS14}. 
For each \emph{interaction} among components in a BIP system, the interaction model of that system specifies the set of ports synchronized by that interaction and the way data is retrieved, filtered and updated in each of
the participating components.
In the third layer, priorities impose scheduling constraints to
resolve conflicts in case alternative interactions are possible. 
In the rest of this paper, we disregard priorities and focus mainly on
interaction models (cf., footnote \ref{fn:prio}).

\paragraph{Data-agnostic semantics.} 
We first introduce a data-agnostic semantics for BIP.

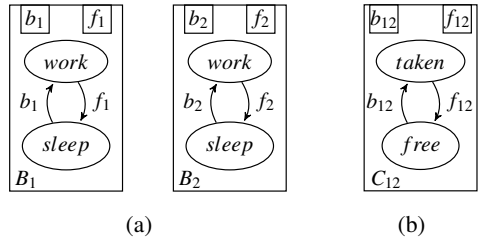
\begin{figure}[t]
\centering
\newcommand{\SCALE}{.75}
\subfigure[]{
\scalebox{\SCALE}{\begin{tikzpicture}[->,>=stealth',shorten >=1pt,auto,node distance=1.5cm, semithick]
  \tikzstyle{every state}=[fill=red,draw=none,text=white]

  \node[elliptic state]         (A)              {$sleep$};
  \node[elliptic state]         (B) [above of=A] {$work$};

  \path (A) edge [bend left]  node {$b_1$} (B);
  \path (B) edge [bend left]  node {$f_1$} (A);
  \draw (-1,-0.8) rectangle (1,2.5);
  \draw (0.3,2) rectangle (0.8,2.5);
  \draw (-0.8,2) rectangle (-0.3,2.5);
  \node at (-0.55,2.25) {$b_1$};
  \node at (0.55,2.25) {$f_1$};
  \node at (-0.7,-0.6) {$B_1$};
\end{tikzpicture}
\qquad
\begin{tikzpicture}[->,>=stealth',shorten >=1pt,auto,node distance=1.5cm, semithick]
  \tikzstyle{every state}=[fill=red,draw=none,text=white]

  \node[elliptic state] (A)              {$sleep$};
  \node[elliptic state]         (B) [above of=A] {$work$};

  \path (A) edge [bend left]  node {$b_2$} (B);
  \path (B) edge [bend left]  node {$f_2$} (A);
  \draw (-1,-0.8) rectangle (1,2.5);
  \draw (0.3,2) rectangle (0.8,2.5);
  \draw (-0.8,2) rectangle (-0.3,2.5);
  \node at (-0.55,2.25) {$b_2$};
  \node at (0.55,2.25) {$f_2$};
  \node at (-0.7,-0.6) {$B_2$};
  \label{fig:mutexa}
\end{tikzpicture}}}
\qquad
\subfigure[]{
\scalebox{\SCALE}{\begin{tikzpicture}[->,>=stealth',shorten >=1pt,auto,node distance=1.5cm, semithick]
  \tikzstyle{every state}=[fill=red,draw=none,text=white]

  \node[elliptic state] (A)              {$free$};
  \node[elliptic state]         (B) [above of=A] {$taken$};

  \path (A) edge [bend left,left]  node {$b_{12}$} (B);
  \path (B) edge [bend left,right]  node {$f_{12}$} (A);
  \draw (-1,-0.8) rectangle (1,2.5);
  \draw (0.4,2) rectangle (0.9,2.5);
  \draw (-0.9,2) rectangle (-0.4,2.5);
  \node at (-0.65,2.25) {$b_{12}$};
  \node at (0.65,2.25) {$f_{12}$};
  \node at (-0.6,-0.6) {$C_{12}$};
  \label{fig:mutexb}
\end{tikzpicture}}}

\caption{BIP components (a); coordinator (b).}
\label{fig:mutex}
\end{figure}

\begin{definition}[BIP component~\cite{ABBJS14}]
	\label{defn:bipcomp}
	A {\em BIP component} $C$ over a set of ports $P_C$ is a labeled transition system $(Q,q^0,P_C,\to)$ over the alphabet $2^{P_C}$. If $\calC$ is a set of components, we say that $\calC$ is \emph{disconnected} iff $P_C \cap P_{C'} = \emptyset$ for all distinct $C,C' \in \calC$. Furthermore, we define $P_\calC = \bigcup_{C \in \cC} P_C$.
\end{definition}

Then, BIP defines an \emph{interaction model} over a set of ports $P$ to be a set of subsets of $P$. 
Interaction models are used to define synchronisations among
components, which can be intuitively described as follows.  Given a
disconnected set of BIP components $\mathcal{C}$ and an interaction
model $\gamma$ over $P_\mathcal{C}$, the state space of the
corresponding {\em composite component} $\gamma(\mathcal{C})$ is the
cross product of the state spaces of the components in $\mathcal{C}$;
$\gamma(\mathcal{C})$ can make a transition labelled by an interaction
$N \in \gamma$ iff all the involved components (those that have ports
in $N$) can make the corresponding transitions.  A straightforward
formal presentation can be found in \cite{BliSif07-acp-emsoft} (cf., \defn{archapp} below).
Thus, BIP interaction models are \emph{stateless}: every interaction in $\gamma$ is always allowed; it is enabled if all ports in the interaction are ready.
However, \cite{ABBJS14} shows the need for statefull interaction, which motivates {\em BIP architectures}

\begin{definition}[BIP architecture~\cite{ABBJS14}]
\label{defn:archi}
A {\em BIP architecture} is a tuple $A = (\cC, P_A, \gamma)$, where $\cC$ is a finite disconnected set of {\em coordinating} BIP components, $P_A$ is a set of ports, such that $P_\calC = \bigcup_{C \in \cC} P_C \subseteq P_A$, and $\gamma \subseteq 2^{P_A}$ is a {\em data-agnostic interaction model}. We call ports in $P_A \setminus P_\calC$ {\em dangling ports} of $A$.
\end{definition}

Essentially, a BIP architecture is a structured way of combining an interaction model $\gamma$ with a set of distinguished components, whose only purpose is to control which interactions in $\gamma$ are applicable at which point in time (which depends on the states of the coordinating components).

\begin{definition}[BIP architecture application~\cite{ABBJS14}]
   \label{defn:archapp}
   Let $A = (\cC, P_A, \gamma)$ be a BIP architecture, and $\calB$ a set of components, such that $\calB \cup \calC$ is finite and disconnected, and that $P_A \subseteq P_\calB \cup P_\calC$. Write $\calB \cup \calC = \{B_i \mid i \in I\}$, with $B_i = (Q_i,q_i^0, P_i, \to_i)$. Then, the {\em application $A(\calB)$ of $A$ to $\calB$} is the BIP component $(\prod_{i \in I} Q_i, (q_i)_{i \in I}, P_\calB \cup P_\calC, \to)$, where $\to$ is the smallest relation satisfying: $(q_i)_{i \in I} \xrightarrow{N} (q_i')_{i \in I}$ whenever
\begin{enumerate}
	\item $N = \emptyset$, and there exists an $i\in I$ such that $q_i \xrightarrow{\emptyset}_i q_i'$ and $q_j'=q_j$ for all $j \in I \setminus \{i\}$; or
	\item $N \cap P_A \in \gamma$, and for all $i \in I$ we have $N \cap P_i \neq \emptyset$ implies $q_i \xrightarrow{N \cap P_i}_i q_i'$, and $N \cap P_i = \emptyset$ implies $q_i'=q_i$.
\end{enumerate}
\end{definition}

The application $A(\cB)$, of a BIP architecture $A$ to a set of BIP components $\cB$, enforces coordination constraints specified by that architecture on those components~\cite{ABBJS14}.
The \emph{interface} $P_A$ of $A$
contains all ports $P_\calC$ of the coordinating components $\cC$ and some
additional ports, which must belong to the components in $\cB$.  In
the application $A(\cB)$, the ports belonging to $P_A$ can
participate only in interactions defined by the interaction model
$\gamma$ of $A$. 
Ports that do not belong to $P_A$ are not restricted and can participate in any interaction.  

Intuitively, an architecture can also be viewed as an incomplete
system: the application of an architecture consists in ``attaching''
its dangling ports to the operand components.  The operational
semantics is that of composing all components (operands and
coordinators) with the interaction model as described in the previous
paragraph.  The intuition behind transitions labelled by $\emptyset$
is that they represent {\em observable idling} (as opposed to internal
transitions).  This allows us to ``desynchronise'' combined
architectures (see \defn{archcomp}) in a simple manner, since
coordinators of one architecture can idle, while those of another 
performs a transition.  Note that, if $N = \emptyset$, in item 2
of \defn{archapp}, $N \cap P_i = \emptyset$, hence also, $q_i' = q_i$, for all
$i$.  Thus, intuitively, one can say that none of the components
moves.  Item 1, however, does allow one component to make a real move
labelled by $\emptyset$, if such a move exists.  Thus, the transitions
labelled by $\emptyset$ interleave, reflecting the idea that in BIP
synchronisation can happen only through ports.

\begin{example}[Mutual exclusion\cite{ABBJS14}]
  \label{ex:mutex:base}
  Consider the components $B_1$ and $B_2$ in \fig[(a)]{mutex}.  In order to
  ensure mutual exclusion of their \texttt{work} states, we apply the
  BIP architecture $A_{12} = (\{C_{12}\}, P_{12}, \gamma_{12})$, where $C_{12}$
  is shown in \fig[(b)]{mutex}, $P_{12} = \{b_1, b_2, b_{12}, f_1, \allowbreak f_2,
  f_{12}\}$ and $\gamma_{12} = \bigl\{\emptyset, \{b_1, b_{12}\}, \allowbreak \{b_2, b_{12}\}, \{f_1,
  f_{12}\}, \{f_2, f_{12}\}\bigr\}$.  
  The interface $P_{12}$ of $A_{12}$ covers all ports of $B_1$, $B_2$ and
  $C_{12}$.  Hence, the only possible interactions are those that explicitly
  belong to $\gamma_{12}$.  
  Assuming that the initial states of $B_1$ and $B_2$ are \texttt{sleep},
  and that of $C_{12}$ is \texttt{free}, neither
  of the two states $(\mathtt{free}, \mathtt{work}, \mathtt{work})$ and
  $(\mathtt{taken}, \mathtt{work}, \mathtt{work})$ is reachable, \ie the
  mutual exclusion property $(q_1 \neq \mathtt{work}) \lor (q_2 \neq 
  \mathtt{work})$\mdash 
  where $q_1$ and $q_2$ are state variables of $B_1$ and $B_2$ respectively\mdash holds in  $A_{12}(B_1, B_2)$.
\tri\end{example}

\begin{definition}[Composition of BIP architectures~\cite{ABBJS14}]
\label{defn:archcomp}
Let $A_1 = (\calC_1, P_1, \gamma_1)$ and $A_2 = (\calC_2, P_2, \gamma_2)$ be two BIP architectures.
Recall that $P_{\calC_i} = \bigcup_{C \in \cC_i} P_C$, for $i=1,2$. 
If $P_{\calC_1} \cap P_{\calC_2} = \emptyset$, then $A_1 \oplus A_2$ is given by $(\calC_1 \cup \calC_2 ,P_1 \cup P_2, \gamma_{12})$, where $\gamma_{12} = \{ N \subseteq P_1 \cup P_2 \mid N \cap P_i \in \gamma_i, \mbox{ for } i=1,2\}$.
In other words, $\gamma_{12}$ is the interaction model defined by the conjunction of the characteristic predicates of $\gamma_1$ and $\gamma_2$.
\end{definition}

\paragraph{Data-aware semantics.} Recently, the data-agnostic formalization of BIP interaction models was extended with data transfer, using the notion of \emph{interaction expressions}~\cite{BBJS14}.
Let $\calP$ be a global set of ports.
For each port $p \in \calP$, let $x_p\oftype\domD_p$ be a typed variable used for the data exchange at that port. 
For a set of ports $P \subseteq \calP$, let $X_P = (x_p)_{p \in P}$.
An interaction expression models the effect of an interaction among ports in terms of the data exchanged through their corresponding variables.

\begin{definition}[Interaction expression~\cite{BBJS14}]
  \label{defn:expression}
  An {\em interaction expression} is an expression of the form
  {
    \[
    \connector{P \leftarrow Q}{g(X_Q,X_L)}{(X_P,X_L) := \up(X_Q,X_L)}
              {(X_Q,X_L) := \down(X_P, X_L)}\,,
              \]
  }
  where
  $P,Q \subseteq \calP$ are {\em top} and {\em bottom} sets of ports;
  $L \subseteq \calP$ is a set of {\em local} variables;
  $g(X_Q,X_L)$ is the boolean {\em guard};
  $\up(X_Q,X_L)$ and $\down(X_P,X_L)$ are respectively the {\em up-} and {\em
    downward data transfer} expressions.

  For an interaction expression $\alpha$ as above, we define by
  $\head{\alpha} \bydef{=} P$, $\tail{\alpha} \bydef{=} Q$ and
  $\support{\alpha} \bydef{=} P \cup Q$ the sets of top, bottom and all
  ports in $\alpha$, respectively.  We denote $g_\alpha$, $\up_\alpha$ and
  $\down_\alpha$ the guard, upward and downward transfer corresponding expressions in $\alpha$.
\end{definition}

The first part of an interaction expression, $(P \leftarrow Q)$, describes
the control flow as a dependency relation between the bottom and
the top ports.  The expression in the brackets describes the data flow, first ``upward''---from bottom to top ports---and then ``downward''.
The guard $g(X_Q, X_L)$ relates these two parts:
interaction is enabled only when the values of the local variables together
with those of variables associated to the bottom ports satisfy a boolean
condition. As a side effect, an interaction expression may also
modify local variables in $X_L$.
Intuitively, such an interaction expression can \emph{fire} only if its guard is true.
When it fires, its upstream transfer is computed first using the values offered by its participating BIP components.  Then, the downstream
transfer modifies all the port variables with updated values.

\begin{definition}[BIP interaction models~\cite{BBJS14}]
	A {\em (data-aware) BIP interaction model} is a set $\Gamma$ of {\em simple BIP connectors} $\alpha$, which are BIP interaction expressions of the form
\[ \connector{\{w\} \leftarrow A}{g(X_A)}{(x_w,X_L) := \up(X_A)}{X_A := \down(x_w, X_L)}, \]
where $w \in P$ is a single top port, $A \subseteq P$ is a set of ports, such that $w \not\in A$, and neither $\up$ nor $g$ involves local variables.
\end{definition}

\begin{example}[Maximum]
\label{ex:maximum}
	Let $\calP = \{a,b,w,l\}$ be a set of ports of type integer, i.e., $x_p \oftype \domD_p = \mathbb{Z}$, for all $p \in \calP$, and consider the interaction expression (simple BIP connector)
\[ \alpha_{\max} = \connector{\{w\} \leftarrow \{a,b\}}{{\tt tt}}{x_l := \max(x_a,x_b)}{x_a,x_b := x_l}, \]
where ${\tt tt}$ is true. First, the connector takes the values presented at ports $a$ and $b$. Then, the simple BIP connector $\alpha_{\max}$ computes atomically the maximum of $x_a$ and $x_b$ and assigns it to its local variable $x_l$. Finally, $\alpha_{\max}$ assigns atomically the value of $x_l$ to both $x_a$ and $x_b$.  
\tri\end{example}

BIP interaction expressions capture complete information about all aspects
of component interaction---i.e. synchronisation and data transfer
possibilities---in a structured and concise manner.  Thus, by
examining interaction expressions, one can easily understand, on the one
hand, the interaction model used to compose components and, on the
other hand, how the valuations of data variables affect the
enabledness of the interactions and how these valuations are modified.
Furthermore, a formal definition of a composition operator on
interaction expressions is provided in \cite{BBJS14}, which allows combining
such expressions hierarchically to manage the complexity of systems
under design.  Since any BIP system can be flattened, this
hierarchical composition of interaction expressions is not relevant
for the semantic comparison of BIP and Reo in this paper.
Nevertheless, the possibility of concisely capturing all aspects of
component interaction in one place is rather convenient.

\subsection{Reo}
\label{sec:reo}

Reo is a coordination language wherein graph like structures express concurrency constraints (e.g., synchronization, exclusion, ordering, etc.) among multiple components. 
These structures consist of a composition of channels and nodes, collectively called \emph{connectors} or \emph{circuits}. 
A channel in Reo has exactly two \emph{ends}, and each end either accepts data items, if it is a \emph{source end}, or offers data items, if it is a \emph{sink end}. 
Moreover, a channel has a \emph{type} for its behaviour in terms of a formal constraint on the dataflow through its two ends. 
Its abstract definition of channels and its notion of channel types make Reo an extensible programming language.
Beside the established channel types (\tab{channels} contains some of them) Reo allows arbitrary user-defined channel types.

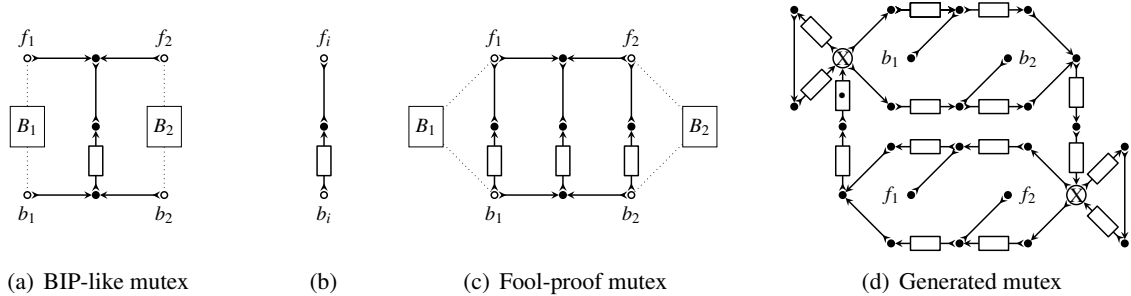
\begin{figure}[t]
\centering
\subfigure[BIP-like mutex]{\scalebox{.7}{\begin{tikzpicture}[baseline, node distance=1.3cm, every node/.style={transform shape}]
        \node[reobnode,label=above:$f_1$]  (c1) [] {};
        \node[reonode]  (c2) [right of=c1] {};
        \node[reobnode,label=above:$f_2$]  (c3) [right of=c2] {};
        \node[]  (c4) [below of=c1] {};
        \node[reonode]  (c5) [below of=c2] {};
        \node[]  (c6) [below of=c3] {};       
        \node[reobnode,label=below:$b_1$]  (c7) [below of=c4] {};
        \node[reonode]  (c8) [below of=c5] {};
        \node[reobnode,label=below:$b_2$]  (c9) [below of=c6] {};
        \node[]  (f1) [left of=c7] {};
        \node[]  (f2) [right of=c9] {};        
        \node[rectangle, draw, text centered, minimum height=2em] (B1) [below of=c1] {$B_1$};
        \node[rectangle, draw, text centered, minimum height=2em] (B2) [below of=c3] {$B_2$};
        \draw[sync,>->] (c1) to node {} (c2);
        \draw[sync,>->] (c3) to node {} (c2);
        \draw[sync,>->] (c7) to node {} (c8);
        \draw[sync,>->] (c9) to node {} (c8);
        \draw[syncdrain] (c2) to node {} (c5);
        \draw[fifo,>->] (c8) to node {} (c5);
        \draw[dotted, -] (B1) to node {} (c1);
        \draw[dotted, -] (B1) to node {} (c7);
        \draw[dotted, -] (B2) to node {} (c3);
        \draw[dotted, -] (B2) to node {} (c9);
        
        \node[] (above) [above of=c1] {};
        \node[] (below) [below of=c7] {};
\end{tikzpicture}
\label{fig:reomutexa}}}
\qquad
\subfigure[]{\scalebox{.7}{\begin{tikzpicture}[baseline, node distance=1.3cm, every node/.style={transform shape}]
        \node[reobnode,label=above:$f_i$]  (c1) {};
        \node[reonode]  (c2) [below of=c1] {};
        \node[reobnode,label=below:$b_i$]  (c3) [below of=c2] {};
        \draw[syncdrain] (c1) to node {} (c2);
        \draw[fifo,>->] (c3) to node {} (c2);
        
        \node[] (above) [above of=c1] {};
        \node[] (below) [below of=c3] {};
\end{tikzpicture}
\label{fig:reomutexb}}}
\qquad
\subfigure[Fool-proof mutex]{\scalebox{.7}{\begin{tikzpicture}[baseline, node distance=1.3cm, every node/.style={transform shape}]
        \node[]  (b1) [] {};
        \node[reobnode,label=above:$f_1$]  (c1) [right of=b1] {};
        \node[reonode]  (c2) [right of=c1] {};
        \node[reobnode,label=above:$f_2$]  (c3) [right of=c2] {};
        \node[]  (b2) [right of=c3] {};   
        \node[reonode]  (c4) [below of=c1] {};
        \node[reonode]  (c5) [below of=c2] {};
        \node[reonode]  (c6) [below of=c3] {};       
        \node[reobnode,label=below:$b_1$]  (c7) [below of=c4] {};
        \node[reonode]  (c8) [below of=c5] {};
        \node[reobnode,label=below:$b_2$]  (c9) [below of=c6] {};
        \node[]  (f1) [left of=c7] {};
        \node[]  (f2) [right of=c9] {};        
        \node[rectangle, draw, text centered, minimum height=2em] (B1) [below of=b1] {$B_1$};
        \node[rectangle, draw, text centered, minimum height=2em] (B2) [below of=b2] {$B_2$};
        \draw[sync,>->] (c1) to node {} (c2);
        \draw[sync,>->] (c3) to node {} (c2);
        \draw[sync,>->] (c7) to node {} (c8);
        \draw[sync,>->] (c9) to node {} (c8);
        \draw[dotted, -] (B1) to node {} (c1);
        \draw[dotted, -] (B1) to node {} (c7);
        \draw[dotted, -] (B2) to node {} (c3);
        \draw[dotted, -] (B2) to node {} (c9);
        \draw[syncdrain] (c1) to node {} (c4);
        \draw[syncdrain] (c2) to node {} (c5);
        \draw[syncdrain] (c3) to node {} (c6);
        \draw[fifo,>->] (c7) to node {} (c4);
        \draw[fifo,>->] (c8) to node {} (c5);
        \draw[fifo,>->] (c9) to node {} (c6);
        
        \node[] (above) [above of=c1] {};
        \node[] (below) [below of=c7] {};
\end{tikzpicture}
\label{fig:reomutexc}
}}
\qquad
\subfigure[Generated mutex]{\scalebox{.7}{\begin{tikzpicture}[baseline, node distance=1.3cm, every node/.style={transform shape}]
	\node[draw,thick,circle, label={[label distance=-0.46cm]0:X}] (a1) [] {};
	\node[reonode] (a2) [above right of=a1] {};
	\node[reonode] (a3) [right of=a2] {};
	\node[reonode] (a4) [right of=a3] {};
	\node[reonode] (a5) [below right of=a4] {};
	\node[reonode] (a6) [below right of=a1] {};
	\node[reonode] (a7) [right of=a6] {};
	\node[reonode] (a8) [right of=a7] {};
	\node[reonode] (a9) [below of=a5] {};
	\node[draw,thick,circle, label={[label distance=-0.46cm]0:X}] (a11) [below of=a9] {};
	\node[reonode] (a12) [above left of=a11] {};
	\node[reonode] (a13) [left of=a12] {};
	\node[reonode] (a14) [left of=a13] {};
	\node[reonode] (a15) [below left of=a14] {};
	\node[reonode] (a16) [below left of=a11] {};
	\node[reonode] (a17) [left of=a16] {};
	\node[reonode] (a18) [left of=a17] {};
	\node[reonode] (a19) [above of=a15] {};
	\node[reonode] (a20) [above left of=a1] {};
	\node[reonode] (a21) [below left of=a1] {};
	\node[reonode] (a22) [above right of=a11] {};
	\node[reonode] (a23) [below right of=a11] {};
	
	\node[reonode,label=left:$b_1$] (b1) [below left of=a3] {};
	\node[reonode,label=right:$b_2$] (b2) [above right of=a7] {};
	\node[reonode,label=left:$f_1$] (f1) [below left of=a13] {};
	\node[reonode,label=right:$f_2$] (f2) [above right of=a17] {};
	
	\draw[syncdrain] (a3) to node {} (b1);
	\draw[syncdrain] (a7) to node {} (b2);
	\draw[syncdrain] (a13) to node {} (f1);
	\draw[syncdrain] (a17) to node {} (f2);
	
	\draw[fifo,>->] (a2) to node {} (a3);
	\draw[fifo,>->] (a3) to node {} (a4);
	\draw[fifo,>->] (a6) to node {} (a7);
	\draw[fifo,>->] (a7) to node {} (a8);
	\draw[fifo,>->] (a12) to node {} (a13);
	\draw[fifo,>->] (a13) to node {} (a14);
	\draw[fifo,>->] (a16) to node {} (a17);
	\draw[fifo,>->] (a17) to node {} (a18);
	\draw[fifo,>->] (a5) to node {} (a9);
	\draw[fifo,>->] (a9) to node {} (a11);
	\draw[fifo,>->] (a15) to node {} (a19);
	\draw[fifofull,>->] (a19) to node {} (a1);
	\draw[fifo,>->] (a2) to node {} (a3);
	\draw[fifo,>->] (a2) to node {} (a3);
	\draw[fifo,>->] (a2) to node {} (a3);
	
	\draw[fifo,>->] (a1) to node {} (a20);
	\draw[fifo,>->] (a21) to node {} (a1);
	\draw[fifo,>->] (a11) to node {} (a22);
	\draw[fifo,>->] (a23) to node {} (a11);
	
	\draw[sync,>->] (a1) to node {} (a2);
	\draw[sync,>->] (a1) to node {} (a6);
	\draw[sync,>->] (a4) to node {} (a5);
	\draw[sync,>->] (a8) to node {} (a5);
	\draw[sync,>->] (a11) to node {} (a12);
	\draw[sync,>->] (a11) to node {} (a16);
	\draw[sync,>->] (a14) to node {} (a15);
	\draw[sync,>->] (a18) to node {} (a15);
	\draw[sync,>->] (a20) to node {} (a21);
	\draw[sync,>->] (a22) to node {} (a23);
        
    \node[] (below) [below of=a11] {};
	
\end{tikzpicture}
\label{fig:generatedmutex}}}
\caption{Fool-proof (c) mutual exclusion protocol in Reo, composed from a BIP-like (a) mutual exclusion connector and an altenator connector (b), and the generated Reo circuit (d) from \ex{mutextranslation}.}
\label{fig:reomutex}
\end{figure}

Multiple ends may glue together into \emph{nodes} with a fixed \emph{merge-replicate} behaviour: a data item out of a single sink end coincident on a node, atomically propagates to all source ends coincident on that node. This propagation happens only if all their respective channels allow the data exchange. 
A node is called a \emph{source node} if it consists of source ends, a \emph{sink node} if it consists of sink ends, and a \emph{mixed node} otherwise. 
Together, the source and sink nodes of a connector constitute its set of \emph{boundary nodes/ports}.

\begin{example}
\label{ex:reomutex}
\fig{reomutexa} shows a Reo connector that achieves mutual exclusion of components $B_1$ and $B_2$, exactly as the BIP system shown in \fig{mutex} does. This connector consists of a composition of channels and nodes in \tab{channels}. The Reo connector atomically accepts data from either $b_1$ or $b_2$ and puts it into the {\tt FIFO1} channel, a buffer of size one.
A full {\tt FIFO1} channel means that $B_1$ or $B_2$ holds the lock. 
If one of the components writes to $f_1$ or $f_2$, the {\tt SyncDrain} channel flushes the buffer, and the lock is released, returning the connector to its initial configuration, where $B_1$ and $B_2$ can again compete for exclusive access by attempting to write to $b_1$ or $b_2$.

Note that this connector is not fool-proof. Even if $B_1$ takes the lock, $B_2$ may release it, and vice versa. 
Hence, exactly as the BIP architecture in \fig{mutex}, the Reo connector in \fig{reomutexa} relies on the conformance of the coordinated components $B_1$ and $B_2$. 
The expected behaviour of $B_i$, $i=1,2$, is that it alternates writes on the $b_i$ and $f_i$, and that every write on $f_i$ comes after a write on $b_i$.
Depending on such assumptions may not be ideal.
The connector, shown in \fig{reomutexb}, makes this expected behaviour explicit. 
By composing two such connectors with the connector in \fig{reomutexa}, we obtain a fool-proof mutual exclusion protocol, as shown in \fig{reomutexc}. 
\fig{foolproofmutex} shows the constraint automaton semantics of the connector in \fig{reomutexc}.
Unlike the case of the connector in \fig{reomutexa} or the BIP architecture in \fig{mutex}, non-compliant writes to $b_i$ or $f_i$ ports of the connector in \fig{reomutexc} will {\em block} component $B_i$, but cannot {\em break} the mutual exclusion protocol that this connector implements.
\tri
\end{example}

\paragraph{Formal semantics of Reo.} 
Reo has a variety of formal semantics \cite{Arbab11, JA12}. In this paper we use its operational \emph{constraint automaton} (CA) semantics \cite{BSAR06}.

\begin{definition}[Constraint automata \cite{BSAR06}] 
\label{defn:CA} 
Let $\calN$ be a set of nodes and $\calD$ a set of data items. A data constraint is a formula in the language of the grammar
\[ g \to \top \mid \neg g \mid g \wedge g \mid \exists d_p (g) \mid d_p = v, \quad \mbox{ with } p \in \calN, v \in \calD, \]
where variable $d_p$ represents the data assigned to (i.e., exchanged through) port $p$.
Let $\models$ denote the obvious satisfaction relation between data constraints and data assignments $\delta : N \to \calD$, with $N \subseteq \calN$, and write $DC(\calN, \calD)$ for the set of all data constraints.
A constraint automaton (over data domain $\calD$) is a tuple
$\calA  = (Q,\calN,\to,q_0)$
where $Q$ is a set of states, $\calN$ is a finite set of nodes, ${\to} \subseteq Q \times 2^\calN \times DC(\calN, \calD) \times Q$ is a transition relation, and $q_0 \in Q$  is the initial state.
\end{definition}

In this paper, we consider only finite data domains, although most of our results generalize to infinite data domains. Over a finite data domain, the data constraint language $DC(\calN,\calD)$ is expressive enough to define any data assignment. 
For notational convenience, we relax, in this paper, the definition of data constraints and allow the use of set-membership and functions in the data constraints. However, we preserve the intention that a data constraint describes a set of data assignments.

\tab{channels} shows the CA semantics for some typical Reo primitives.
The CA semantics of every Reo connector can be derived as a composition of the constraint automata of its primitives, using the CA product operation in \defn{ProductCA}.
On the other hand, every constraint automaton (over a finite data domain) translates back into a Reo connector~\cite{BKK14}.
Because of this correspondence, we may consider Reo and CA as equivalent, and focus on constraint automata only.
\begin{table}[t]
\begin{center}
\begin{tabular}{ccccc}
{\tt Sync} 
& 
{\tt LossySync} 
& 
{\tt SyncDrain} 
&
{\tt FIFO1} 
&
{\tt Node} 
\\
\hline
\adjustbox{valign=c}{\scalebox{.8}{\begin{tikzpicture}
  \node[label=above:$A$] (in) at (0,0) {};
  \node[right of= in, xshift=1cm,label=above:$B$] (out) {};
  \draw[sync,>->] (in) -- (out); 
\end{tikzpicture}}}
&
\adjustbox{valign=c}{\scalebox{.8}{\begin{tikzpicture}
  \node[label=above:$A$] (in)                {};
  \node[right of=in, xshift=1cm,label=above:$B$] (out)  {};
  \draw[lossysync,>->] (in) -- (out);
\end{tikzpicture}}}
& 
\adjustbox{valign=c}{\scalebox{.8}{\begin{tikzpicture}
  \node[label=above:$A$] (in)                {};
  \node[right of=in, xshift=1cm,label=above:$A'$] (out)  {};
  \draw[syncdrain,>-<] (in) -- (out);
\end{tikzpicture}}}
& 
\adjustbox{valign=c}{\scalebox{.8}{\begin{tikzpicture}
  \node[label=above:$A$] (in) at (0,0) {};
  \node[right of=in, xshift=1cm,label=above:$B$] (out)  {};
  \draw[fifo,>->] (in) -- (out);
\end{tikzpicture}}}
&
\adjustbox{valign=c}{\scalebox{.8}{\begin{tikzpicture}
  \node[reobnode] (A) at (0,0) [] {};
  \node[above right of=A, node distance=.7cm, inner sep=1pt] (A1) {$B$};
  \node[above left of=A, node distance=.7cm, inner sep=1pt] (A2) {$A$};
  \node[below left of=A, node distance=.7cm, inner sep=1pt] (A3) {$B'$};
  \node[below right of=A, node distance=.7cm, inner sep=1pt] (A4) {$A'$};
  \draw[sync,->] (A1) -- (A);
  \draw[sync,>-] (A) -- (A2);
  \draw[sync,->] (A3) -- (A);
  \draw[sync,>-] (A) -- (A4);
\end{tikzpicture}}} 
\\ 
\\
\adjustbox{valign=c}{\scalebox{.8}{\begin{tikzpicture}[->,initial text={},>=stealth',shorten >=1pt,auto,node distance=1.5cm, semithick,scale=0.8, every node/.style={transform shape}]
  \tikzstyle{every state}=[fill=red,draw=none,text=white]

  \node[state] (A)              {$q$};

  \path (A) edge [in=110,out=70,loop above]  node {$\{A,B\},\top$} (A);
\end{tikzpicture}}}
&
\adjustbox{valign=c}{\scalebox{.8}{\begin{tikzpicture}[->,initial text={},>=stealth',shorten >=1pt,auto,node distance=1.5cm, semithick,scale=0.8, every node/.style={transform shape}]
  \tikzstyle{every state}=[fill=red,draw=none,text=white]

  \node[state] (A)              {$q$};

  \path (A) edge [in=110,out=70,loop above]  node {$\{A,B\},\top$} (A);
  \path (A) edge [in=290,out=250,loop below]  node {$\{A\},\top$} (A);
\end{tikzpicture}}}
&
\adjustbox{valign=c}{\scalebox{.8}{\begin{tikzpicture}[->,initial text={},>=stealth',shorten >=1pt,auto,node distance=1.5cm, semithick,scale=0.8, every node/.style={transform shape}]
  \tikzstyle{every state}=[fill=red,draw=none,text=white]

  \node[state] (A)              {$q$};

  \path (A) edge [in=110,out=70,loop above]  node {$\{A,A'\},\top$} (A);
\end{tikzpicture}}}
&
\adjustbox{valign=c}{\scalebox{.8}{\begin{tikzpicture}[->,initial text={},>=stealth',shorten >=1pt,auto,node distance=1.5cm, semithick,scale=0.8, every node/.style={transform shape}]
  \tikzstyle{every state}=[fill=red,draw=none,text=white]

  \node[state,initial] (A)              {$q_0$};
  \node[state] (B) [right of=A] {$q_1$};

  \path (A) edge [bend left=70] node [above] {$\{A\},\top$} (B);
  \path (B) edge [bend left=70] node [below] {$\{B\},\top$} (A);
\end{tikzpicture}}}
&
\adjustbox{valign=c}{\scalebox{.8}{\begin{tikzpicture}[->,initial text={},>=stealth',shorten >=1pt,auto,node distance=1.5cm, semithick,scale=0.8, every node/.style={transform shape}]
  \tikzstyle{every state}=[fill=red,draw=none,text=white]

  \node[state] (A)              {$q$};

  \path (A) edge [in=110,out=70,loop above] node {$\{B,A,A'\},\top$} (A);
  \path (A) edge [in=290,out=250,loop below]  node {$\{B',A,A'\},\top$} (A);
\end{tikzpicture}}}
\\
\hline
\end{tabular}
\end{center}
\caption{Some primitives in the Reo language with CA semantics over a singleton data domain $\calD$.}
\label{tab:channels}
\end{table}

If a constraint automaton $\calA$ has only one state, $\calA$ is called \emph{stateless}. 
If the data domain $\calD$ of $\calA$ is a singleton, $\calA$ is called a \emph{port automaton} \cite{KC09}. In that case, we omit data constraints, because all satisfiable constraints reduce to $\top$.

\begin{definition}[Product of CA \cite{BSAR06}]
\label{defn:ProductCA} 
Let $\calA_i = (Q_i,\calN_i,\to_i,q_{0,i})$ be a constraint automaton, for $i=1,2$. Then the product $\calA_1 \Join \calA_2$ of these automata is the automaton $(Q_1 \times Q_2,\calN_1 \cup \calN_2,\to,(q_{0,1},q_{0,2}))$, whose transition relation is the smallest relation obtained by the rule: $(q_1,q_2) \xrightarrow{N_1\cup N_2, g_1 \wedge g_2} (q_1',q_2')$ whenever
\begin{enumerate}
	\item $q_1 \xrightarrow{N_1,g_1}_1 q_1'$, $q_2 \xrightarrow{N_2,g_2}_2 q_2'$, and $N_1 \cap \calN_2 = N_2 \cap \calN_1$, or
	\item $q_i \xrightarrow{N_i,g_i}_i q_i'$, $N_j = \emptyset$, $g_j = \top$, $q_j' = q_j$, and $N_i \cap \calN_j = \emptyset$ with $j \in \{1,2\} \setminus \{i\}$.
\end{enumerate}
\end{definition}

It is not hard to see that constraint automata product operator is associative and commutative modulo equivalence of state names and data constraints.

\begin{definition}[Hiding in CA \cite{BSAR06}]
\label{defn:HidingCA} 
Let $\calA = (Q,\calN,\to,q_0)$ be a constraint automaton, and $P=\{p_1,\ldots,p_n\}$ a set of nodes. Then hiding nodes $P$ of $\calA$ yields an automaton $\exists P (\calA) = (Q,\calN \setminus P,\to_\exists,q_0)$, where $\to_\exists$ is given by $\{ (q,N \setminus P,\exists d_{p_1} \cdots \exists d_{p_n} (g),q') \mid (q,N,g,q') \in \ \to\}$.
\end{definition}

The hiding operator affects only transition labels, and preserves the structure of the automaton. 
Hence the hiding operator offers a technique to alter the interface of a component or connector without modifying its behaviour. 
As hiding of non-shared nodes distributes over the product, hiding of non-shared nodes commutes with constraint automata product.

\begin{example}[Product and hide]
\label{ex:prodhide}
Consider the Reo connectors in \fig{reomutex}. Using \defn{ProductCA}, and the primitive constraint automata from \tab{channels}, we find their CA semantics as shown in Figures \ref{fig:biplikemutex}, \ref{fig:alternator}, and \ref{fig:foolproofmutex}, respectively. 
If we compute the product of the automaton $\calA_0$ in \fig{biplikemutex} with the automata $\calA_i$, $i=1,2$, in \fig{alternator}, then we obtain an automaton $\calA$, whose part reachable from the initial state $(0,0,0)$ is shown in \fig{foolproofmutex}.
\tri\end{example}

\section{Port automata and BIP architectures}
\label{sec:PA2Arch}

To study the relation between BIP and Reo with respect to synchronization, we start by defining a correspondence between them in the data-agnostic domain. This correspondence consists of a pair of mappings between the sets containing semantic models of BIP and Reo connectors.
For the data independent semantic model of Reo connectors we choose port automata: a restriction of constraint automata over a singleton set as data domain. We model BIP connectors by BIP architectures introduced in \cite{ABBJS14}. 
In order to compare the behaviour of BIP and Reo connectors we interpret them as labeled transition systems. 
We define a mapping $\Reoa$ that transforms BIP architectures into port automata, and a mapping $\BIPa$ that transforms port automata into BIP architectures. 
We then show that these mappings preserve (1) properties closed under bisimulation, and (2) composition structure modulo semantic equivalence.

\subsection{Interpretation of BIP and Reo}
\label{sec:interpr_no_data}
To compare the behaviour of BIP and Reo connectors, we interpret all connectors as labeled transitions systems with one initial state and an alphabet $2^P$, for a set of ports $P$. 
We write $\LTS$ for the class of all such labeled transition systems. 

\begin{figure}[t]
\centering
\subfigure[data-agnostic domain]{\scalebox{1}{
\begin{tikzpicture}[baseline={([yshift=-.5ex]current bounding box.center)}]
  \matrix (m) [matrix of math nodes,row sep=1.5em,column sep=4em,minimum width=2em,ampersand replacement=\&]
  {
     \mathrm{Reo} \&  \& \mathrm{BIP} \\
     \PA \&  \& \Arch \\
     \& \LTS \& \\};
  \draw[-latex] (m-2-1) to node [below left] {$\fa$} (m-3-2);
  \draw[-latex] (m-2-1.10) to node [above] {$\BIPa$} (m-2-3.170);
  \draw[-latex] (m-2-3) edge node [below right] {$\ga$} (m-3-2);
  \draw[-latex] (m-2-3.190) to node [below] {$\Reoa$} (m-2-1.350);
  \draw[-latex] (m-1-1) to [bend left=10] node [right] {\cite{BSAR06}} (m-2-1);
  \draw[-latex] (m-2-1) to [bend left=10] node [left] {\cite{BKK14}} (m-1-1);
  \draw[-] (m-1-3.260) to node [right] {\cite{ABBJS14}} (m-2-3.100);
  \draw[-] (m-2-3.80) to node [] {} (m-1-3.280);
\end{tikzpicture}
\label{fig:intpr1}
}}
\qquad
\subfigure[data-sensitive domain]{\scalebox{1}{
\begin{tikzpicture}[baseline={([yshift=-.5ex]current bounding box.center)}]
  \matrix (m) [matrix of math nodes,row sep=1.5em,column sep=4em,minimum width=2em,ampersand replacement=\&]
  {
     \mathrm{Reo} \&  \& \mathrm{BIP} \\
     \CA \&  \& \BC \\
     \& \LTS \& \\};
  \draw[-latex] (m-2-1) to node [below left] {$\fb$} (m-3-2);
  \draw[-latex] (m-2-1.10) to node [above] {$\BIPb$} (m-2-3.170);
  \draw[-latex] (m-2-3) edge node [below right] {$\gb$} (m-3-2);
  \draw[-latex] (m-2-3.190) to node [below] {$\Reob$} (m-2-1.350);
  \draw[-latex] (m-1-1) to [bend left=10] node [right] {\cite{BSAR06}} (m-2-1);
  \draw[-latex] (m-2-1) to [bend left=10] node [left] {\cite{BKK14}} (m-1-1);
  \draw[-] (m-1-3.260) to node [right] {\cite{BBJS14}} (m-2-3.100);
  \draw[-] (m-2-3.80) to node [] {} (m-1-3.280);
\end{tikzpicture}
\label{fig:intpr2}
}}
\caption{Translations and interpretations in data-agnostic and data-sensitive domain.}
\label{fig:intpr}
\end{figure}
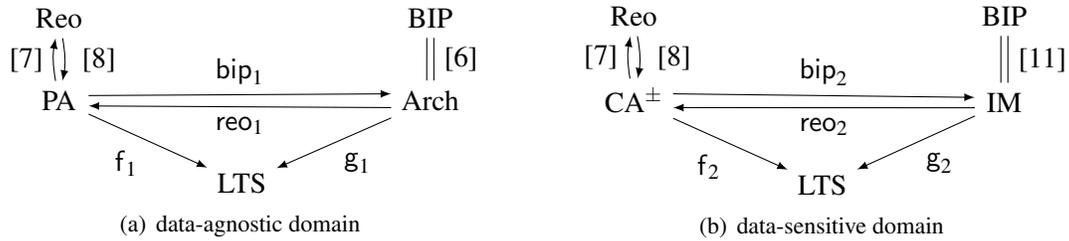

\fig{intpr1} shows our translations and interpretations. The objects $\PA$, $\Arch$ and $\LTS$ are, respectively, the classes
of port automata, BIP architectures, and labeled transition systems. The mappings $\BIPa$, $\Reoa$, $\fa$ and $\ga$, respectively, translate Reo to BIP, BIP to Reo, Reo to LTS, and BIP to LTS.

We first consider the semantics of connectors. Since BIP connectors differ internally from Reo connectors, we restrict our interpretation to their observable behaviour. This means that we hide the ports of the coordinating components in BIP architectures. For port automata this means that for our comparison, we implicitly assume that all names represent boundary nodes.

The interpretation of a port automaton in LTS is defined by
\begin{equation} 
\label{eqn:fa}
\fa((Q,\calN, \to, q_0)) = (Q, 2^\calN, \to, q_0). 
\end{equation}
Hence $\fa$ acts essentially as an identity function, justifying our choice of interpretation. 
Next, we define the interpretation of BIP architectures using their operational semantics obtained by applying them on dummy components and hiding all internal ports. 
Let $A = (\calC, P, \gamma)$ be a BIP architecture with coordinating components $\calC = \{C_1,\ldots,C_n\}$, $n \geq 0$, and $C_i = (Q_i,q_i^0,P_i,\to_i)$. 
Recall that $P_\calC = \bigcup_i P_i$ is the set of internal ports in $A$. Define 
$D = (\{q_D\},q_D, P, \{ (q_D,N,q_D) \mid \emptyset \neq N \subseteq P \setminus P_\calC\})$ as a dummy component relative to the BIP architecture $A$.
Using \defn{archapp}, we compute the BIP architecture application $A(\{D\}) = ((\prod_{i=1}^n Q_i) \times \{q_D\},(\vec{q}^0,q_D),P,\to_s)$ of $A$ to its dummy component $D$. 
Then, 
\begin{equation} 
\label{eqn:ga} 
\textstyle \ga(A) = (\prod_{i=1}^n Q_i\times \{q_D\},2^{P \setminus P_\calC},\{((\vec{q},q_D),N \setminus P_\calC,(\vec{q}',q_D)) \mid (\vec{q},q_D) \xrightarrow{N}_s (\vec{q}',q_D) \},(\vec{q}^0,q_D)) 
\end{equation}
In other words, $\ga(A)$ equals $A(\{D\})$ after hiding all internal ports $P_\calC$.
Note that we based our interpretation $\ga$ on the operational semantics of BIP architectures, i.e., BIP architecture application. This justifies the definition of interpretation of architectures.

Because of hiding, $\ga$ is not injective. Hence, our interpretation of BIP architectures induces a non-trivial equivalence given by equality of interpretations. In the sequel, we use a slightly stronger version of equivalence based on bisimulation \cite{Milner89}. 

\begin{definition}[Bisimulation~\cite{Milner89}]
\label{defn:bisim}
If $L_i = (Q_i,2^{P_i}, \to_i,q^0_i) \in \LTS$, $i=1,2$, then $L_1$ and $L_2$ are \emph{bisimilar} ($L_1 \cong L_2$) iff $P_1 = P_2$ and there exists $R \subseteq Q_1 \times Q_2$ such that $(q^0_1,q^0_2) \in R$, and $(q_1,q_2) \in R$ implies, for all $N \in 2^{P_i}$, $i,j \in \{1,2\}$ with $i\neq j$, if $q_i \xrightarrow{N}_i q_i'$, then, for some $q_j'$, $q_j \xrightarrow{N}_j q_j'$ and $(q_1',q_2') \in R$.
\end{definition}

\begin{definition}[Semantic equivalence]
	Let $\calA,\calB \in \PA$ be port automata and $A,B \in \Arch$ be BIP architectures. Then, $\calA$ and $\calB$ are {\em semantically equivalent} ($\calA \sim \calB$) iff $\fa(\calA) \cong \fa(\calB)$, and $A$ and $B$ are {\em semantically equivalent} ($A \sim B$) iff $\ga(A) \cong \ga(B)$.
\end{definition}

With a common semantics for BIP and Reo, we can define the notion of preservation of properties expressible in this common semantics. Recall that a property of labeled transition systems corresponds to the subset of labeled transition systems satisfying that property. 

\begin{definition}
\label{defn:prop-data-agnostic}
Let $P \subseteq \LTS$ be a property. Then, \emph{$\BIPa$ preserves $P$} iff $\fa(\calA) \in P \Leftrightarrow \ga(\BIPa(\calA)) \in P$ for all $\calA \in \PA$. Similarly, \emph{$\Reoa$ preserves $P$} iff $\ga(A) \in P \Leftrightarrow \fa(\Reoa(A)) \in P$ for all $A \in \Arch$.
\end{definition}

\subsection{BIP to Reo}
\label{sec:bip2reo_a}
To translate BIP connectors to Reo connectors, we first determine what elements of BIP architectures correspond to Reo connectors. Our interpretations of port automata and BIP architectures show that dangling ports in BIP architectures correspond to boundary port names in port automata. Furthermore, the mutual exclusion of the interactions in an interaction model in a BIP architecture simulates mutually exclusive firing of transitions in port automata. The definition of a coordinating component in a BIP architecture is almost identical to that of a port automaton, yielding an obvious translation.

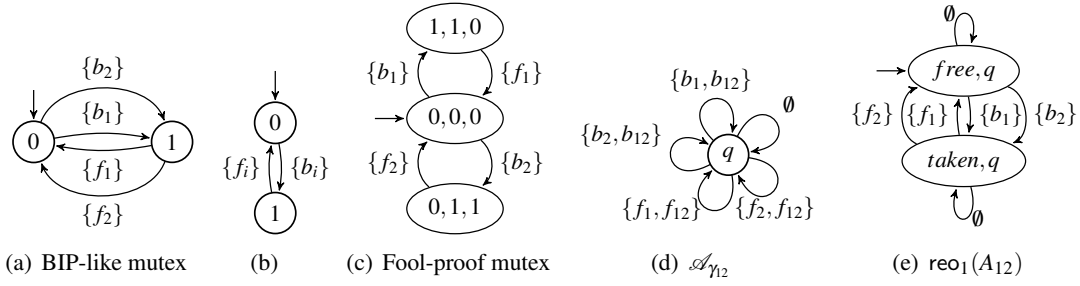
\begin{figure}[t]
\centering
\subfigure[BIP-like mutex]{
\scalebox{.8}{
\begin{tikzpicture}[->, initial text={}, initial where=above,>=stealth',shorten >=1pt,auto,node distance=1.5cm, semithick]
  \tikzstyle{every state}=[fill=red,draw=none,text=white]

  \node[state, initial] (A)              {$0$};
  \node[state] (B) [right of=A,xshift=0.8cm] {$1$};

  \path (A) edge [bend left=10,above] node {$\{b_1\}$} (B);
  \path (A) edge [bend left=70,above] node {$\{b_2\}$} (B);
  \path (B) edge [bend left=10,below] node {$\{f_1\}$} (A);
  \path (B) edge [bend left=70,below] node {$\{f_2\}$} (A);
  \label{fig:biplikemutex}
\end{tikzpicture}}}
\subfigure[]{
\scalebox{.8}{
\begin{tikzpicture}[->, initial text={}, initial where=above,>=stealth',shorten >=1pt,auto,node distance=1.5cm, semithick]
  \tikzstyle{every state}=[fill=red,draw=none,text=white]

  \node[state,initial] (A)              {$0$};
  \node[state] (B) [below of=A] {$1$};

  \path (A) edge [bend left=10,right] node {$\{b_i\}$} (B);
  \path (B) edge [bend left=10,left] node {$\{f_i\}$} (A);
  \label{fig:alternator}
\end{tikzpicture}}}
\subfigure[Fool-proof mutex]{
\scalebox{.8}{
\begin{tikzpicture}[->, initial text={}, initial where=left,>=stealth',shorten >=1pt,auto,node distance=1.5cm, semithick]
  \tikzstyle{every state}=[fill=red,draw=none,text=white]

  \node[elliptic state,initial] (A)              {$0,0,0$};
  \node[elliptic state] (B) [above of=A] {$1,1,0$};
  \node[elliptic state] (C) [below of=A] {$0,1,1$};

  \path (A) edge [bend left=50,left]  node {$\{b_1\}$} (B);
  \path (A) edge [bend left=50,right]  node {$\{b_2\}$} (C);
  \path (B) edge [bend left=50,right]  node {$\{f_1\}$} (A);
  \path (C) edge [bend left=50,left]  node {$\{f_2\}$} (A);
  \label{fig:foolproofmutex}
\end{tikzpicture}}}
\subfigure[$\calA_{\gamma_{12}}$]{
\scalebox{.8}{
 \begin{tikzpicture}[->, initial text={}, initial where=right,>=stealth',shorten >=1pt,auto,node distance=1.5cm, semithick]
  \tikzstyle{every state}=[fill=red,draw=none,text=white]
  \node[state]         (A)              {$q$};
  \path (A) edge [in=8,out=64,loop,above right]  node {$\emptyset$} (A);
  \path (A) edge [in=80,out=136,loop,above]  node {$\{b_1,b_{12}\}$} (A);
  \path (A) edge [in=152,out=208,loop,above left]  node {$\{b_2,b_{12}\}$} (A);
  \path (A) edge [in=224,out=280,loop,left]  node {$\{f_1,f_{12}\}$} (A);
  \path (A) edge [in=296,out=352,loop,below]  node {$\{f_2,f_{12}\}$} (A);
  \label{fig:mutexaReo}
\end{tikzpicture}}}
\subfigure[$\Reoa(A_{12})$]{
\scalebox{.8}{
\begin{tikzpicture}[->, initial text={}, initial where=left,>=stealth',shorten >=1pt,auto,node distance=1.5cm, semithick]
  \tikzstyle{every state}=[fill=red,draw=none,text=white]

  \node[elliptic state,initial] (A)              {$free,q$};
  \node[elliptic state] (B) [below of=A] {$taken,q$};

  \path (A) edge [loop above,left]  node {$\emptyset$} (A);
  \path (B) edge [loop below,right]  node {$\emptyset$} (B);
  \path (A) edge [bend left=10,right]  node {$\{b_1\}$} (B);
  \path (A) edge [bend left=70,right]  node {$\{b_2\}$} (B);
  \path (B) edge [bend left=10,left]  node {$\{f_1\}$} (A);
  \path (B) edge [bend left=70,left]  node {$\{f_2\}$} (A);
  \label{fig:mutexbReo}
\end{tikzpicture}}}
\caption{CA representations $(a)$, $(b)$, and $(c)$ of Reo connectors Figures \ref{fig:reomutexa}, \ref{fig:reomutexb}, and \ref{fig:reomutexc}, respectively; translation of the interaction model $(d)$ and BIP architecture $(e)$ of Figure \ref{fig:mutex}.}
\label{fig:mutexReo}
\end{figure}

Let $A = (\calC, P, \gamma)$ be a BIP architecture, with $\calC = \{C_1,\ldots,C_n\}$. Each $C_i$ corresponds trivially to a port automaton $\widetilde{C_i}$. Let $\calA_\gamma = (\{q\}, P, \to, q)$ be the stateless port automaton over $P$ with transition relation $\to$ defined by $ \{ (q, N, q) \mid N \in \gamma \}$. Then $\calA_\gamma$ can be seen as the port automata encoding of the interaction model $\gamma$. Recall that $P_\calC = \bigcup_{C \in \calC} P_C$. The corresponding port automaton of $A$ is given by 
\begin{equation} 
\label{eqn:ReoPA} 
\Reoa(A) = \exists P_\calC (\widetilde{C_1} \Join \cdots \widetilde{C_n} \Join \calA_\gamma).
\end{equation}

\begin{example}
\label{ex:mutextranslation}
We translate the BIP architecture in \ex{mutex:base} using \eq{ReoPA}. 
First, we transform $\gamma_{12}$ into a port automaton $\calA_{\gamma_{12}}$, shown in \fig{mutexaReo}. Then, we compute the product of $\calA_{\gamma_{12}}$ with the coordinating component $C_{12}$ to obtain the port automaton corresponding to the BIP architecture $A_{12}$, shown in \fig{mutexbReo}.
As mentioned in section \secn{reo}, we can transform the port automaton in \fig{mutexbReo} into a Reo connector, using the method described in \cite{BKK14}.
This mechanical translation yields the Reo connector in \fig{generatedmutex}. 
Here, the dot in the {\tt FIFO1} buffer indicates that its initial state is the full state.
The crossed node represents an {\em exclusive router}, which atomically takes data from a coincident sink end, and provides it to a single coincident source end.
Note that the port automaton semantics of the connector in \fig{reomutexa} (see \fig{biplikemutex}) is similar to the automaton in \fig{mutexbReo}, up to empty transitions.
\tri\end{example}

\subsection{Reo to BIP}

In BIP, interaction is memoryless. This means that a stateful channel in Reo must translate to a coordinating component. In fact, we may encode the whole Reo connector as one such component. 

Let $\calA_i$, $i=1,2$, be two port automata, and let $p \in \calN_1 \cap \calN_2$ be a shared port of $\calA_1$ and $\calA_2$. Suppose that we know how to translate $\calA_i$ into a BIP architecture $A_i$. If $p$ is not a dangling port of $A_1$, then, by symmetry, $p$ is not a dangling port of $A_2$. But now, $A_1$ and $A_2$ are not composable, because there components are not disconnected. Hence, since we want the translation to preserve composition, $p$ should be a dangling port.

Let $\calA = (Q, \calN, \to, q_0)$ be a port automaton. We construct a corresponding BIP architecture. Duplicate all ports in $\calN$ by defining $N' = \{ n' \mid n \in N \}$ for all $N \subseteq \calN$. We do not use a port $n'$, for $n \in \calN$, for composition. Their exact name is therefore not important, but merely their relation to its dangling brother $n$. Trivially, $\overline{\calA} = (Q, q_0, \calN', \to_c)$, with $\to_c \ = \{ (q,N',q') \mid (q,N,q') \in \ \to \}$, is a BIP component (cf., \defn{bipcomp}). Essentially, $\calA$ and $\overline{\calA}$ are the same labeled transition system. Now we define:
\begin{equation}
\label{eqn:BIPArch}
\BIPa(\calA) = (\{\overline{\calA}\}, \calN \cup \calN', \{N \cup N'\mid N \subseteq \calN\}).
\end{equation}
Thus, $\BIPa$ uses the port automaton as the coordinating component of the generated BIP architecture.

\begin{example}
Let $\calA$ be the port automaton in \fig{alternator} over the name set $\calN = \{ b_i,f_i\}$. We determine $\BIPa(\calA)$. Obtain $\overline{\calA}$ by adding adding a prime to each port in $\calA$. The interaction model of $\BIPa(\calA)$ consist of $\{N \cup N'\mid N \subseteq \calN\} = \bigl\{\emptyset, \{b_i,b_i'\}, \{f_i,f_i'\},\{b_i,b_i', f_i,f_i'\} \bigr\}.$ Hence, $\BIPa(\calA)$ is given by th BIP architecture $(\{\overline{\calA}\}, \{b_i,f_i,b_i',f_i'\}, \bigl\{\emptyset, \{b_i,b_i'\}, \{f_i,f_i'\},\{b_i,b_i', f_i,f_i'\} \bigr\} )$.
\end{example}

\subsection{Preservation of properties}

To confirm that translations $\Reoa$ and $\BIPa$ preserve properties, we first investigate whether \fig{intpr1} commutes, i.e., $\fa(\Reoa(A)) = \ga(A)$ and $\ga(\BIPa(\calA)) = \fa(\calA)$, for $A \in \Arch$ and $\calA \in \PA$.

First, note that the equations $\fa(\Reoa(A)) = \ga(A)$ and $\ga(\BIPa(\calA)) = \fa(\calA)$ cannot hold, because their state spaces differ. For example, $\ga$ alters the state space by adding the state of a dummy component, and $\Reoa$ adds the state of the port automaton encoding of the interaction model. Therefore we view these equations modulo bisimulation of labeled transition systems from \defn{bisim}.

Next, consider the equation $\fa(\Reoa(A)) \cong \ga(A)$, for some BIP architecture $A = (\{C_1,\ldots,C_n\},P,\gamma)$. Suppose that two distinct coordination components $C_i$ and $C_j$, $1 \leq i < j \leq n$, each contains an empty-labeled transition, i.e., there exist transistions $(q_i,\emptyset,q_i') \in \ \to_i$ and $(q_j,\emptyset,q_j') \in \ \to_j$. When we translate $A$ to a port automaton using $\Reoa$, the second rule in \defn{ProductCA} yields a \emph{single} transition in $\fa(\Reoa(A))$ from a global state where component $C_i$ is in state $q_i$ and $C_j$ is in state $q_j$, to a global state where $C_i$ is in state $q_i'$ and $C_j$ is in state $q_j'$. However, BIP semantics does not allow independent progress of state-changing empty-labeled transitions, which means that this single transition exists only when $q_i' = q_i$ and $q_j' = q_j$. Indeed, the first rule of \defn{archapp} allows either $C_i$ or $C_j$ to change state, and the second rule implies $q_i' = q_i$ and $q_j' = q_j$ for $N = \emptyset$. Because of this, we need to exclude BIP architectures where two coordinating components can make a state-changing empty-labeled transition. Moreover, as we consider composition of BIP architectures in \secn{compatcomposition}, we exclude BIP architectures containing a single coordinating component that can make a state-changing empty-labeled transition, and restrict $\Arch$ to 
$\Arch' = \{ A \in \Arch \mid \forall C_i \in \calC \, : \, q_i \xrightarrow{\emptyset}_i q_i' \Rightarrow q'_i = q_i\}$.
Finally, consider the equation $\ga(\BIPa(\calA)) \cong \fa(\calA)$, for some port automaton $\calA$. Note that the interaction model of $\BIPa(\calA)$ contains the empty set. Hence, the second rule in \defn{archapp} yields empty-labeled self-transitions in $\ga(\BIPa(\calA))$. Since $\fa$ acts like the identity, we conclude that $\calA$ should have empty-labeled self-transitions, i.e., $q' = q$ implies $(q ,\emptyset, q') \in \ \to$. On the other hand, suppose that $(q ,\emptyset, q') \in \ \to$. Then the coordinating component of $\BIPa(\calA)$ should not contain a state-changing empty-labeled transition, hence $q' = q$. Therefore, we restrict $\PA$ to $\PA' = \{ \calA \in \PA \mid q \xrightarrow{\emptyset} q' \Leftrightarrow q' = q \}$.

\begin{theorem}
\label{thm:commute_a} 
For all $\calA \in \PA'$ and $A \in \Arch'$ we have $\ga(\BIPa(\calA)) \cong \fa(\calA)$ and $\fa(\Reoa(A)) \cong \ga(A)$.
\end{theorem}

\begin{proof}
Using \defn{archapp}, \defn{ProductCA}, $A \in \Arch'$, $\calA \in \PA'$, and the fact that $(q_D,\emptyset,q_D) \notin \ \to_D$, it follows that (1) $\sim$ given by $(q,q_D) \sim q$ for all $q \in Q$ is a bisimulation between $\ga(\BIPa(\calA))$ and $\fa(\calA)$, where $Q$ is the state space of $\calA$, and (2) $\approx$ given by $(\vec{q},q_I) \approx (\vec{q},q_D)$ for all $\vec{q}=(q_i)_{i \in I} \in \prod_{i \in I} Q_i$, is a bisimulation, where $Q_i$, $i \in I$, are the state spaces of the coordinating components of $A$. 
See \cite{bip2reo} for a detailed proof.
\end{proof}

\begin{corollary}
\label{cor:props_pa2arch}
$\BIPa$ and $\Reoa$ preserve all properties closed under bisimulation, i.e., for all $P \subseteq \LTS$, $\calA \in \PA'$ and $A \in \Arch'$ we have $\fa(\calA) \in P \Leftrightarrow \ga(\BIPa(\calA)) \in P$ and $\ga(A) \in P \Leftrightarrow \fa(\Reoa(A)) \in P$.
\end{corollary}

\begin{example}
	Consider the following safety property $\varphi$ satisfied by the Reo connector in \fig{reomutexc}: ``if $b_1$ fires, then $b_2$ fires only after $f_1$ fires''. Clearly, the automaton $\calA'$, obtained from \fig{foolproofmutex} by adding empty self-transitions, satisfies this property as well. Using \cor{props_pa2arch}, we conclude that the BIP architecture $\BIPa(\calA) = \BIPa(\calA')$ satisfies $\varphi$. More generally, \cor{props_pa2arch} allows model checking of BIP architectures with Reo model checkers. 
\tri\end{example}

\subsection{Compatibility with composition} 
\label{sec:compatcomposition}

BIP architectures and port automata have their own notions of composition.
This raises the question of whether our translations preserve composition structures. 
We show that, under specific conditions, our translations preserve composition modulo semantic equivalence. Recall the port automaton representation of the interaction model (\secn{bip2reo_a}).

\begin{lemma}
\label{lem:interaction}
Let $A_i = (\calC_i, P_i, \gamma_i) \in \Arch$, $i=1,2$, with $P_{\calC_1} \cap P_{\calC_2} = \emptyset$ and $\emptyset \in \gamma_1 \cap \gamma_2$.  
Then, we have that $\calA_{\gamma_{12}} \sim \calA_{\gamma_1} \Join \calA_{\gamma_2}$,
where $\gamma_{12}$ be the interaction model of $A_1 \oplus A_2$.
\end{lemma}

\begin{proof}
Follows easily from \defn{ProductCA} and \defn{archcomp}. See \cite{bip2reo} for a detailed proof.
\end{proof}

Suppose that $\Reoa(A_1 \oplus A_2) \sim \Reoa(A_1) \Join \Reoa(A_2)$, for any two BIP architectures $A_1,A_2 \in \Arch'$. \defn{ProductCA} implies $N_{\Reoa(A_1 \oplus A_2)} = N_{\Reoa(A_1)\Join \Reoa(A_2)} = N_{\Reoa(A_1)} \cup N_{\Reoa(A_2)}$. In other words, the name set of port automaton $\Reoa(A_1 \oplus A_2)$ is the union of the name set of the port automata $\Reoa(A_i)$, $i=1,2$. Hence, $N_{\Reoa(A_i)} \subseteq N_{\Reoa(A_1 \oplus A_2)}$, for $i=1,2$. This means that the dangling ports of $\Reoa(A_1 \oplus A_2)$ contain all dangling ports of $\Reoa(A_i)$. Therefore, we need to assume that $P_{\calC_1} \cap P_2 = P_{\calC_2} \cap P_1 = \emptyset$.

Note that this is only a mild assumption. Indeed, if $p \in P_{\calC_1} \cap P_2$ is a dangling port of $P_2$, connected directly to a component in $A_1$. Then, we first add a (dangling) port $x$ to $A_1$ and synchronize $p$ with $p'$ by considering the BIP interaction model $\gamma_1' = \{ N \cup \{x\} \mid p \in N \in \gamma_1 \} \cup \{ N \mid p \notin N \in \gamma \}$. Finally, we rename $p$ to $x$ in $A_2$. The resulting architectures satisfy the assumption.

\begin{theorem}
\label{thm:homreo} 
$\Reoa(A_1 \oplus A_2) \sim \Reoa(A_1) \Join \Reoa(A_2)$ for all $A_i = (\calC_i,P_i,\gamma_i) \in \Arch'$, with $P_{\calC_1} \cap P_2 = P_{\calC_2} \cap P_1 = \emptyset$ and $\emptyset \in \gamma_1 \cap \gamma_2$.
\end{theorem}

\begin{proof}
Let $\calC_1 \cup \calC_2 = \{C_1,\ldots,C_n,\ldots,C_{m}\}$, with $C_i \in \calC_1$ iff $i \leq n$. By definition, we have $\Reoa(A_1 \oplus A_2) = \exists P_{\calC_1 \cup \calC_2} (\widetilde{C_1} \Join \cdots \widetilde{C_n}\Join \widetilde{C_{n+1}} \Join \cdots \widetilde{C_m} \Join \calA_{\gamma_{12}})$. Next, we use the bisimulation of port automata (i.e., constraint automata with data contraint $\top$) as defined in \cite{BSAR06}. Composition ($\Join$) of port automata is commutative and associative up to bisimulation \cite{BSAR06}. Using \lem{interaction}, it follows that $\Reoa(A_1 \oplus A_2) \cong \exists P_{\calC_1} \exists P_{\calC_2} (\widetilde{C_1} \Join \cdots \widetilde{C_n} \Join \calA_{\gamma_1} \Join \widetilde{C_{n+1}} \Join \cdots \widetilde{C_m}\Join\calA_{\gamma_2})$.
Indeed, since $\fa$ is like the identity, it follows that semantic equivalence $\sim$ coincides with bisimulation $\simeq$ of port automata as defined in \cite{BSAR06}.
Now, we use our assumption that $P_{\calC_1} \cap P_2 = P_{\calC_2} \cap P_1 = \emptyset$, and the fact that $\widetilde{C_1},\ldots, \widetilde{C_n}$, and $\calA_{\gamma_1}$ do not use ports from $P_{\calC_2}$. Then, $\Reoa(A_1 \oplus A_2) \cong \exists P_{\calC_1} (\widetilde{C_1} \Join \cdots \widetilde{C_n} \Join \calA_{\gamma_1}) \Join \exists P_{\calC_2} (\widetilde{C_{n+1}} \Join \cdots \widetilde{C_m}\Join \calA_{\gamma_2}))$.
We conclude that $\Reoa(A_1 \oplus A_2) \cong \Reoa(A_1) \Join \Reoa(A_2)$.
Since, $\fa$ is like the identity, it is not hard to see that $\fa$ takes bisimilar port automata to bisimilar labeled transition systems. Therefore, $\Reoa$ is a homomorphism up to semantic equivalence, i.e., $\Reoa(A_1 \oplus A_2) \sim \Reoa(A_1) \Join \Reoa(A_2)$.
\end{proof}

\begin{theorem}
\label{thm:hombip}
$\BIPa(\calA_1 \Join \calA_2) \sim \BIPa(\calA_1) \oplus \BIPa(\calA_2)$ for all $\calA_i \in \PA'$.
\end{theorem}

\begin{proof}
Note that, since $\fa$ is like the identity, semantic equivalence $\sim$ coincides with bisimulation $\simeq$ of port automata \cite{BSAR06}. As $\simeq$ is a congruence with respect to the composition $\Join$ of port automata, we conclude that $\sim$ is a congruence too (i.e., $\fa(\calA_i) \cong \fa(\calA_i')$, for $i=1,2$, implies $\fa(\calA_1 \Join \calA_2) \cong \fa(\calA_1' \Join \calA_2')$). 

Let $\calA_i\in \PA'$, $i=1,2$, be two port automata. From \thm{homreo}, we conclude that
$\fa(\Reoa(A_1 \oplus A_2)) \cong \fa(\Reoa(A_1) \Join \Reoa(A_2))$, for any $A_1,A_2 \in \Arch'$. Substitute $A_i = \BIPa(\calA_i)$, for $i=1,2$. Then,
$\fa(\Reoa(\BIPa(\calA_1) \oplus \BIPa(\calA_2))) \cong \fa(\Reoa(\BIPa(\calA_1)) \Join \Reoa(\BIPa(\calA_2)))$. Thus, $\fa(\Reoa(\BIPa(\calA_i))) \cong \ga(\BIPa(\calA_i)) \cong \fa(\calA_i)$, for $i=1,2$, by \thm{commute_a}. Hence, using that $\sim$ is a congruence, we obtain $\ga(\BIPa(\calA_1) \oplus \BIPa(\calA_2)) \cong \fa(\calA_1 \Join \calA_2)$. Therefore, $\ga(\BIPa(\calA_1) \oplus \BIPa(\calA_2)) \cong \ga(\BIPa(\calA_1\Join \calA_2))$.
\end{proof}

\begin{example}
	For any two ports $x$ and $y$, let $\calA_{\{x,y\}}$ be the port automaton of a synchronous channel (cf., \tab{channels}), and let $C_{\{x,y\}}$ be its corresponding BIP component. Suppose we need to translate $\calA_{\{a,b\}} \Join \calA_{\{b,c\}}$ to a BIP architecture. Then we first compute $\BIPa(\calA_{\{a,b\}}) = (\{ C_{\{a',b'\}}\}, \{a,a',b,b'\}, \gamma_{\{a,b\}})$, with $\gamma_{\{a,b\}} = \{\emptyset, \{a,a'\},\{b,b'\},\{a,a',b,b'\}\}$. Next, we compute $\BIPa(\calA_{\{b,c\}}) = (\{ C_{\{b'',c''\}}\}, \{b,b'',c,c''\}, \gamma_{\{b,c\}})$, with $\gamma_{\{b,c\}} = \{\emptyset, \{b,b''\},\{c,c''\},\{b,b'',c,c''\}\}$. Note that we need to use a double prime now, because otherwise $b'$ would be a shared port of $C_{\{a',b'\}}$ and $C_{\{b'',c''\}}$. Using \thm{hombip}, we find that $ \BIPa(\calA_{\{a,b\}} \Join \calA_{\{b,c\}} = \BIPa(\calA_{\{a,b\}}) \oplus \BIPa(\calA_{\{b,c\}}) = ( \{ C_{\{a',b'\}}, C_{\{b'',c''\}} \}, \{a,a',b,b',b'',c,c''\}, \gamma_{\{a,b,c\}})$, where $\gamma_{\{a,b,c\}}$ is the composition of $\gamma_{\{a,b\}}$ and $\gamma_{\{b,c\}}$.
\end{example}

\begin{example}
Consider the port automaton $\calA'$, obtained from \fig{foolproofmutex} by adding empty self-transitions. If we translate $\calA'$ to BIP, we obtain a BIP architecture $B_1 = \BIPa(\calA')$, which has only a single coordinating component. From \ex{prodhide} we conclude $\calA' \cong \calA_0' \Join \calA_1' \Join \calA_2'$, where $\calA_0$ is the port automaton in \fig{biplikemutex}, $\calA_i$, $i=1,2$, is the port automaton in \fig{alternator}, and $\calA_i'$ is obtained from $\calA_i$ by adding empty self-transitions. Now consider $B_3 = \BIPa(\calA_0') \oplus \BIPa(\calA_1') \oplus \BIPa(\calA_2')$. Using \defn{archcomp}, we see that $B_3$ has three coordinating components. Nevertheless, \thm{hombip} shows that $B_3$ is semantically equivalent to $B$. Therefore, \thm{hombip} allows to compute translations compositionally.
\tri\end{example}

\section{Stateless CA's and interaction models}
\label{sec:CA2BIPconn}

In \secn{PA2Arch} we established a correspondence between port automata and BIP architectures. Here, we offer translations between data-aware connector models in BIP and Reo.

First we determine the semantic model of the connectors. For BIP connectors we use BIP interaction models, i.e., sets of interaction expressions $\alpha$, with a single top port that is not a bottom port, and whose guard and up functions are independent of local variables (\defn{expression}). We assume that every top port occurs only in one interaction expression per BIP interaction model. We denote the class of BIP interaction models by $\BC$.
For the semantics of Reo connectors we take a pair consisting of a constraint automaton together with a partition of its node set into source nodes $\calN_{src}$, mixed nodes $\calN_{mix}$, and sink nodes $\calN_{snk}$. We call such pairs \emph{constraint automata with polarity}. 
Due to the absence of coordinating components in the data sensitive model for BIP, we restrict ourselves here to stateless constraint automata, since BIP interaction expressions are stateless \cite{ABBJS14, BBJS14}.
We write $\CA$ for the class of all stateless constraint automata with polarity, with $\calN_{src} = P^\ast = \{ p^\ast \mid p \in P \}$ and $\calN_{snk} = P_\ast = \{ p_\ast \mid p \in P \}$ for some set of ports $P$. This assumption is necessary to enable simulation of bidirectional ports in BIP. The reason we explicitly distinguish node types in this semantics is to give direction to dataflow, similar to BIP connectors. Usually such node type distinctions are implicit, but for preciseness we encode them as a partition within the semantics of Reo connectors. 

As in \secn{PA2Arch}, we interpret all connectors as labeled transition systems. Then we define translations between Reo connectors ($\CA$) and BIP connectors ($\BC$), and show that they preserve properties.

\subsection{Interpretation of BIP and Reo}
\label{sec:interpretation2}
An important difference between BIP and Reo involves how they handle data. 
BIP uses bidirectional ports, while Reo treats input and output ports separately.
Since the common semantics should support both approaches, we duplicate every bidirectional port of BIP to obtain two unidirectional ports, compatible with Reo.
The sense of every reference to a bidirectional port in a BIP interaction expression maps that bidirectional port to its intended corresponding unidirectional port.

Let $\LTS$ be the class of all labeled transition systems over an alphabet $(D+1)^{2P}$, where $D$ is a set of data items; $1 = \{0\}$ contains \emph{void} or \emph{null}, modeling the absence of data; and $2P$ is the \emph{duplicated (unidirectional) port set} of a set of (bidirectional) ports $P$, that is, $2P = \{p^\ast, p_\ast \mid p \in P\}$. If data appears at $p^\ast$ (i.e., $\delta(p^\ast) \neq 0$ for $\delta \in (D+1)^{2P}$), then we interpret this as input to the connector. If data appears at $p_\ast$, then we interpret this at output from the connector. 

Consider \fig{intpr2}. Classes $\CA$ and $\BC$ consist of constraint automata with polarity and interaction models. Morphisms $\BIPb$ and $\Reob$ are translations of those classes and $\fb$ and $\gb$ are interpretations in a common $\LTS$ semantics.
We do not intend to redefine the semantics of constraint automata with polarity and of interaction models in this section. Hence, we interpret them using their definitions from \cite{BSAR06, BBJS14}. 

We begin by defining the interpretation of stateless constraint automata with polarity. 
Given a stateless constraint automaton with polarity $\calA$, we first determine the smallest set of bidirectional ports $P$ such that $\calN_{src}^{used} \subseteq P^\ast$ and $\calN_{snk}^{used} \subseteq P_\ast$, where $\calN_{src}^{used}$ and $\calN_{snk}^{used}$ are all source and sink nodes that occur on a transition of $\calA$. Then, we take $2P$ as the port names of $\fb(\calA)$. Finally, we obtain the transitions of $\fb(\calA)$ by replacing every transition labeled with $N,g$ in $\calA$ with a set of transitions labeled with $\delta \in \Delta(N,g)$, where $\Delta(N,g)$ contains all data assignments $\delta : 2P \to \calD + 1$ that satisfy the data constraint $N,g$.
We formalize this as follows. Let $\calA = (\{q\}, \calN_{src},\calN_{mix},\calN_{snk}, \rightarrow, q)$ be a stateless constraint automaton with polarity over a data domain $\calD$. 
Define $\calN_{src}^{used} = \bigcup \{ N \cap \calN_{src} \mid q \xrightarrow{N,g} q \}$, and $\calN_{snk}^{used} = \bigcup \{ N \cap \calN_{snk} \mid q \xrightarrow{N,g} q \}$.
Let $P$ be the smallest set, with $\calN_{src}^{used} \subseteq P^\ast$ and $\calN_{snk}^{used} \subseteq P_\ast$. Define 
\begin{equation} 
\fb(\calA) = (\{q\},(\calD+1)^{2P}, \{(q,\delta,q) \mid q \xrightarrow{N,g} q, \delta \in \Delta(N,g) \}),
\end{equation}
where $\Delta(N,g) = \{ \delta : 2P \to \calD +1 \mid \delta(2P \setminus N) = \{0\}, \delta \models g \}$.
Note that ports in $\calN_{src} \setminus \calN_{src}^{used}$ and $\calN_{snk} \setminus \calN_{snk}^{used}$ are important only for composition, which we do not consider in this paper.

Next, we interpret interaction models $\Gamma$ by a single-state labeled transition system with labels describing all possible dataflows allowed by the guard, and up and down functions of some interaction expression in $\Gamma$.
Before we provide a formal definition, we first introduce some notation. 
For every BIP interaction expression $\alpha$, we write $P_\alpha$ for its bottom ports, $g_\alpha$ for its guard, $up_w^\alpha$ and $up_L^\alpha$ for the restriction of the up function to its top port and its local variables, respectively, and $dn_{bot}^\alpha$ for the restriction of the down function to its bottom ports.
For every BIP interaction model $\Gamma$, we write $P_\Gamma = \bigcup_{\alpha \in \Gamma} P_\alpha$, and $D_\Gamma = \bigcup_{p \in P_\Gamma} \domD_p$, where $\domD_p$ is the data type of port $p$.
For every data assignment $\delta : 2P_\Gamma \to D_\Gamma +1$ we define $\delta_{up}(p) = \delta(p^\ast)$ and $\delta_{dn}(p) = \delta(p_\ast)$, for all $p \in P_\alpha$.
Then, we define
\begin{equation}
\label{eqn:g_data}
\gb(\Gamma) = (\{q\},(D_\Gamma+1)^{2P_\Gamma},\{ (q, \delta,q) \mid \alpha \in \Gamma, \delta \in \Delta(\alpha) \subseteq (D_\Gamma +1)^{2P_\Gamma} \}),
\end{equation}
where $\Delta(\alpha) =\{ \delta \mid \delta(2P_\Gamma \setminus 2P_\alpha) = \{0\}, g_\alpha(\delta_{up}) = \mathtt{tt}, \delta_{dn} = dn_{bot}^\alpha(up_w^{\alpha}(\delta_{up}),up_L^{\alpha}(\delta_{up})) \}$.
Note that we use the value of $up_{w}^\alpha(\delta_{up})$ as a local variable, since we consider only non-hierarchical interaction models.

In \cite{BBJS14}, Bliudze et al. encode BIP interaction models in {\em Top/Bottom components}, i.e., an automaton over interaction expressions together with local variables. Furthermore, they define a semantics for T/B components, which indirectly defines an interpretation of interaction models. Equation (\ref{eqn:g_data}) imitates this interpretation without using Top/Bottom components explicitly.

Now that we defined the interpretation of our objects in $\LTS$, we explore how these translations preserve properties that are expressible in $\LTS$, as we did for their counterparts in \secn{interpr_no_data}.

\subsection{Reo to BIP}
\label{sec:reotobip_nodata}
Since BIP interaction models are stateless, we cannot translate an arbitrary constraint automaton (i.e., Reo connector) into BIP. Interaction models in BIP preclude keeping track of the state of a Reo connector. Hence, the translation of the interaction model of a BIP architecture into a port automaton in \secn{bip2reo_a} inspires us for our translation $\BIPb$.

Let $\calA$ be a stateless constraint automaton over a data domain $\calD$. Since we care only about external behaviour, we first hide all mixed nodes. 
Then, we transform every transition in $\calA$ with label $N,g$ into a simple BIP connector with $N$ as its bottom ports, together with a guard, an up and a down function that mimic the data constraint $g$. We define the corresponding set $\BIPb(\calA)$ of simple BIP connectors by the set of all transformed transitions from $\calA$.

We first define the transformation of transitions into interaction expressions. For every label $N,g$ in automaton $\calA$, we define the simple BIP connector
\[ \alpha(N,g) = (\{w_{N,g}\} \leftarrow P_N).[ g_{src}(X_{src}) \, : \, Y_{snk} := \mathrm{solve}(g,X_{src}) \, // \, X_{snk} := Y_{snk} ],\] 
where $P_N$ is the smallest set satisfying $N \cap \calN_{bnd} \subseteq 2P_N$, $g_{src}$ is any quantifier free formula equivalent to $\exists N \setminus \calN_{src} : g$, the variables $Y_{snk} = \{y_p \mid p \in N \cap \calN_{snk}\}$ are some fresh local variables, and $X_{src} = \{x_p \mid p \in N \cap \calN_{src}\}$ and $X_{snk} = \{x_p \mid p \in N \cap \calN_{snk}\}$ model the input and output values assigned to the bottom ports, and $\mathrm{solve}(g,X_{src})$ returns a vector $Y_{snk}$ satisfying $\exists X_{mix} : g(X_{src},X_{mix},Y_{snk})$. All variables have data type $\calD$ (the data domain of $\calA$), i.e., $x_p \oftype\calD$ for all $p \in \calN$.
Note that the solve function in $\alpha(N,g)$ is not deterministic. However, comparing the solve function to the random function in Figure 4 in \cite{BBJS14}, we see that this generality is justified.
Now, we define $\BIPb$ as follows:
\begin{equation} 
\label{eqn:bip(A)} 
\BIPb(\calA) = \{ \alpha(N,g) \mid (q,N,g,q) \in \ \to \}, 
\end{equation}

\subsection{BIP to Reo}
The correspondence between BIP interaction expressions and automata transitions from \secn{reotobip_nodata}, provides the main idea for the translation of interaction models into stateless constraint automata. If $\Gamma$ is a set of simple BIP connectors, we assign to every $\alpha \in \Gamma$ a transition $\tau_\alpha$ labeled with $N(\alpha),g(\alpha)$, and subsequently construct the stateless constraint automaton consisting of all such $\tau_\alpha$ transitions.

Let $\alpha$ be a simple BIP interaction expression. 
Recall our relaxation on the data constraint language in \secn{overview}, and our notations regarding $\alpha$ in \secn{interpretation2}.
Then, define $N(\alpha) \subseteq 2P_\alpha = \{ p_\ast, p^\ast \mid p \in P_\alpha\}$
where $p^\ast \in N(\alpha)$ iff $\alpha$ assignes data to $p$ in the upward data transfer, and $p_\ast \in N(\alpha)$ iff $\alpha$ assigns data to $p$ in the downward data transfer.
Furthermore, let $D_\ast = (d_{p_\ast})_{p \in P}$, $D^\ast = (d_{p^\ast})_{p \in P}$, and define
\[ \textstyle g(\alpha) \, = \, \bigwedge_{p \in P} d_{p^\ast},d_{p_\ast} \in \domD_p \, \wedge \, g_\alpha(D^\ast) \, \wedge \, D_\ast = dn_{bot}^\alpha(up_w^\alpha(D^\ast), up_L^\alpha(D^\ast)),\] 
Note that $g(\alpha)$ is independent of the top port $w$, as we consider only non-hierarchical connectors.

Let $\Gamma$ be a set of simple BIP connectors. Recall that $P_\Gamma = \bigcup_\alpha P_\alpha$ and $D_\Gamma = \bigcup_{p \in P_\Gamma} \domD_p$. Then, define the constraint automaton $\Reob(\Gamma)$ over $D_\Gamma$ by 
\begin{equation}
\label{eqn:reo(gamma)}
\Reob(\Gamma) = (\{q\}, (P_\Gamma)^\ast, \emptyset, (P_\Gamma)_\ast, \{ (q,N(\alpha),g(\alpha),q) \mid \alpha \in \Gamma \}, q).
\end{equation}

\begin{example}
\label{ex:bip2reodata}
Consider the interaction expression $\alpha_{\max}$ from \ex{maximum}, with the data domains restricted to $\calD = \{0,\ldots,2^{32}-1\}$. We translate the interaction model $\Gamma = \{\alpha_{\max}\}$ using \eq{reo(gamma)}, i.e., we compute $\calA = \Reob(\Gamma)$. Trivially, $\calA$ is stateless. Its set of input ports equals $(P_\Gamma)^\ast = \{a^\ast,b^\ast\}$, and its set of output ports equals $(P_\Gamma)_\ast = \{a_\ast,b_\ast\}$. It has a unique transitions $(q,N,g,q)$, with synchronization contraint $N=\{a^\ast,b^\ast,a_\ast,b_\ast\}$ and guard
$g \ \equiv \ \bigvee_{x,y,z \in \calD \ : \ z = \max(x,y)} (d_{a^\ast} = x \wedge d_{b^\ast} = y \wedge d_{a_\ast} = z \wedge d_{b_\ast} = z)$.
\tri\end{example}

\subsection{Preservation of properties}

To show the faithfulness of translations $\BIPb$ and $\Reob$, we show that interpretations $\fb$ and $\gb$ commute with the translations $\BIPb$ and $\Reob$ in \fig{intpr2}.

\begin{theorem} 
\label{thm:commute_b}
For all $\calA \in \CA$ and all $\Gamma \in \BC$ we have $\gb(\BIPb(\calA)) = \fb(\calA)$ and $\fb(\Reob(\Gamma)) = \gb(\Gamma)$.
\end{theorem}

\begin{proof}[Proof. (Sketch)]
Let $\Gamma \in \BC$ and $\calA \in \CA$. Then, $\Delta(\alpha(N,g)) = \Delta(N,g)$, and $\Delta(N(\alpha),g(\alpha)) = \Delta(\alpha)$, for all $\alpha \in \Gamma$, and all transition labels $N,g$ in $\calA$. From this and the definitions of $\fb$ and $\gb$, we see that $\gb(\BIPb(\calA))) = \fb(\calA)$, and $\fb(\Reob(\Gamma)) = \gb(\Gamma)$, respectively.
\end{proof}

\begin{corollary}
\label{cor:props_ca2setssimplecon}
   The translations $\BIPb$ and $\Reob$ preserve all properties expressible in $\LTS$, i.e., $\fb(\calA) \in P \Leftrightarrow  \gb(\BIPb(\calA)) \in P$ and $\gb(\Gamma) \in P \Leftrightarrow  \fb(\Reob(\Gamma)) \in P$ for all $P \subseteq \LTS$, $\calA \in \CA$ and $\Gamma \in \BC$.
\end{corollary}

\begin{example}
Consider the following safety property $\varphi$ for the interaction expression $\alpha_{\max}$ from \ex{maximum}: ``the value retrieved from port $a$ equals zero''.  Clearly, this safety property does not hold, whenever $a$ or $b$ offers a non-zero integer. 
Note that $\varphi$ depends solely on the interpretation of the interaction model $\Gamma = \{\alpha_{\max}\}$ in $\LTS$, and hence $\varphi$ is expressible in $\LTS$.
Using \cor{props_ca2setssimplecon} we conclude that $\varphi$ is false also for $\calA_{\max} = \Reob(\{\alpha_{\max}\})$.
Thus, we know any executable code generated from the constraint automaton $\calA_{\max}$ does not satisfy $\varphi$. More generally, \cor{props_ca2setssimplecon} allows us to use the Reo compiler to generate correct code for a BIP interaction model.
\tri\end{example}

\section{Conclusions and Future Work}
\label{sec:conclusion}

BIP and Reo find common ground in their stimulation of exogenous system design. 
This means that they force the explicit modeling of coordination constraints. 
A clear and formal separation between coordination (connectors) and computation (components) allows the software architect to analyze the interaction of the components using automated tools. 
The exogenous approach of BIP and Reo contrasts with the endogenous approach supported in process algebra and other languages where coordination is woven into the code of the components. 
For example, process algebra does not supply constructs to enforce the separation of concerns necessary in exogenous coordination \cite{PA01}.

Multiparty synchronization constitutes a fundamental coordination concept in BIP (represented by interactions in a BIP interaction model) and Reo (represented by synchronization constraints in constraint automata). Our translations show that these representations of multiparty synchronization coincide.   

The BIP framework concretely \emph{defines} what separates computation (BIP behaviour) and coordination (BIP interaction), while Reo merely \emph{separates} computation (Reo components) and coordination (Reo connector) structurally. Indeed, Reo does not force a fixed universal definition for computation and coordination in all applications.  Without giving a fixed definition of separation criterion, Reo's structural separation of computation from coordination (i.e., component versus connector) simply means that, while this separation is always important, the distinction between the two is in the eye of the beholder: in different applications, different, or even the same people, may find it convenient to draw the line that separates computation and coordination at different places to suit their needs.  For example, the stateful behavior of a {\tt FIFO} with capacity of 1 strictly places what this entity does in the behaviour layer of BIP, as a (computation) component. In Reo, such stateful components can, of course, be regarded and used as computation as well. However, when deemed appropriate, one can use the same component (i.e., a {\tt FIFO1} channel) in the construction of a Reo connector as well, e.g., to express the stateful, turn-taking interaction between two components, as in \fig{reomutex}.

Our data-agnostic translations allow compositional translation, because their operators distribute over composition modulo semantic equivalence.
On the other hand, our data-sensitive translation scheme does not support incremental translation.
It seems intuitive to translate synchronous Reo channels into BIP interaction expressions. 
However, the directionality inherent in the dataflows of BIP interaction expressions implies that they can compose only hierarchically, whereas the \emph{relational} specification of dataflow constraints in Reo (which manifests itself as data constraints in constraint automata transition labels) allows more expressive composition of dataflows as relational composition of constraints.
This difference restricts the set of the Reo connectors that this scheme can incrementally translate into BIP, as well as the granularity of the sub-connectors that it can translate in one increment: the data constraints on the boundary nodes of every such sub-connector must be locally resolvable into a directional dataflow expression at the level of the sub-connector, in isolation.
In practice, synchronous cycles in a Reo connector must translate as a whole, which scuttles the computational benefit of translating incrementally.

In contrast with the BIP architecture model, the data-sensitive model for BIP does not include coordinating components within the connector \cite{ABBJS14, BBJS14}.
Nevertheless, it seems possible to use the formalization in \cite{BBJS14} to extend BIP architectures of \cite{ABBJS14} with data.
However, extending the current composition operator $\oplus$ to compose data-sensitive BIP architectures does not seem trivial, and we do not know what properties such an extended composition operator can preserve. 

Using the ideas from \secn{PA2Arch}, extending our $\Reob$ translation (\fig{intpr2}) to the domain of postulated data-sensitive BIP architectures seems straight-forward. 
Moreover, it may be possible to extend our translations to mappings that preserve internal ports.  
Such extensions, together with the results from \secn{CA2BIPconn}, effectively promise a property-preserving composition operator for data-sensitive BIP architectures that may also share internal ports.

\bibliography{references}

\begin{thebibliography}{10}
\providecommand{\bibitemdeclare}[2]{}
\providecommand{\surnamestart}{}
\providecommand{\surnameend}{}
\providecommand{\urlprefix}{Available at }
\providecommand{\url}[1]{\texttt{#1}}
\providecommand{\href}[2]{\texttt{#2}}
\providecommand{\urlalt}[2]{\href{#1}{#2}}
\providecommand{\doi}[1]{doi:\urlalt{http://dx.doi.org/#1}{#1}}
\providecommand{\bibinfo}[2]{#2}

\bibitemdeclare{}{biptools}
\bibitem{biptools}
 (\bibinfo{year}{2015}): \emph{\bibinfo{title}{{BIP} toolset}}.
\newblock \urlprefix\url{http://www-verimag.imag.fr/BIP-Tools,93.html}.

\bibitemdeclare{}{reotools}
\bibitem{reotools}
 (\bibinfo{year}{2015}): \emph{\bibinfo{title}{{R}eo toolset}}.
\newblock \urlprefix\url{http://reo.project.cwi.nl/reo/wiki/Tools}.

\bibitemdeclare{article}{Reo}
\bibitem{Reo}
\bibinfo{author}{Farhad \surnamestart Arbab\surnameend} (\bibinfo{year}{2004}):
  \emph{\bibinfo{title}{{R}eo: a channel-based coordination model for component
  composition}}.
\newblock {\sl \bibinfo{journal}{Math. Structures Comput. Sci.}}
  \bibinfo{volume}{14}(\bibinfo{number}{3}), pp. \bibinfo{pages}{329--366},
  \doi{10.1017/S0960129504004153}.

\bibitemdeclare{inproceedings}{Arbab11}
\bibitem{Arbab11}
\bibinfo{author}{Farhad \surnamestart Arbab\surnameend} (\bibinfo{year}{2011}):
  \emph{\bibinfo{title}{Puff, The Magic Protocol}}.
\newblock In: {\sl \bibinfo{booktitle}{Talcott Festschrift}}, {\sl
  \bibinfo{series}{Lecture Notes in Comput. Sci.}} \bibinfo{volume}{7000},
  \bibinfo{publisher}{Springer}, pp. \bibinfo{pages}{169--206},
  \doi{10.1007/978-3-642-24933-4\_9}.

\bibitemdeclare{incollection}{ABCLM09}
\bibitem{ABCLM09}
\bibinfo{author}{Farhad \surnamestart Arbab\surnameend},
  \bibinfo{author}{Roberto \surnamestart Bruni\surnameend},
  \bibinfo{author}{Dave \surnamestart Clarke\surnameend}, \bibinfo{author}{Ivan
  \surnamestart Lanese\surnameend} \& \bibinfo{author}{Ugo \surnamestart
  Montanari\surnameend} (\bibinfo{year}{2009}): \emph{\bibinfo{title}{Tiles for
  {R}eo}}.
\newblock In: {\sl \bibinfo{booktitle}{Proc. of WADT}}, {\sl
  \bibinfo{series}{Lecture Notes in Comput. Sci.}} \bibinfo{volume}{5486},
  \bibinfo{publisher}{Springer Berlin Heidelberg}, pp. \bibinfo{pages}{37--55},
  \doi{10.1007/978-3-642-03429-9\_4}.

\bibitemdeclare{article}{ABBJS14}
\bibitem{ABBJS14}
\bibinfo{author}{Paul \surnamestart Attie\surnameend}, \bibinfo{author}{Eduard
  \surnamestart Baranov\surnameend}, \bibinfo{author}{Simon \surnamestart
  Bliudze\surnameend}, \bibinfo{author}{Mohamad \surnamestart Jaber\surnameend}
  \& \bibinfo{author}{Joseph \surnamestart Sifakis\surnameend}
  (\bibinfo{year}{2014}): \emph{\bibinfo{title}{A General Framework for
  Architecture Composability}} \bibinfo{volume}{8702}, pp.
  \bibinfo{pages}{128--143}.
\newblock \doi{10.1007/978-3-319-10431-7\_10}.

\bibitemdeclare{article}{BKK14}
\bibitem{BKK14}
\bibinfo{author}{Christel \surnamestart Baier\surnameend},
  \bibinfo{author}{Joachim \surnamestart Klein\surnameend} \&
  \bibinfo{author}{Sascha \surnamestart Kl{\"{u}}ppelholz\surnameend}
  (\bibinfo{year}{2014}): \emph{\bibinfo{title}{Synthesis of Reo Connectors for
  Strategies and Controllers}}.
\newblock {\sl \bibinfo{journal}{Fundam. Inform.}}
  \bibinfo{volume}{130}(\bibinfo{number}{1}), pp. \bibinfo{pages}{1--20},
  \doi{10.3233/FI-2014-980}.

\bibitemdeclare{article}{BSAR06}
\bibitem{BSAR06}
\bibinfo{author}{Christel \surnamestart Baier\surnameend},
  \bibinfo{author}{Marjan \surnamestart Sirjani\surnameend},
  \bibinfo{author}{Farhad \surnamestart Arbab\surnameend} \&
  \bibinfo{author}{Jan \surnamestart Rutten\surnameend} (\bibinfo{year}{2006}):
  \emph{\bibinfo{title}{Modeling component connectors in {R}eo by constraint
  automata}}.
\newblock {\sl \bibinfo{journal}{Sci. Comput. Programming}}
  \bibinfo{volume}{61}(\bibinfo{number}{2}), pp. \bibinfo{pages}{75--113},
  \doi{10.1016/j.scico.2005.10.008}.

\bibitemdeclare{inproceedings}{bip06}
\bibitem{bip06}
\bibinfo{author}{Ananda \surnamestart Basu\surnameend}, \bibinfo{author}{Marius
  \surnamestart Bozga\surnameend} \& \bibinfo{author}{Joseph \surnamestart
  Sifakis\surnameend} (\bibinfo{year}{2006}): \emph{\bibinfo{title}{Modeling
  Heterogeneous Real-time Components in {BIP}}}.
\newblock In: {\sl \bibinfo{booktitle}{Proc. of SEFM}},
  \bibinfo{publisher}{ACM}, pp. \bibinfo{pages}{3--12},
  \doi{10.1109/SEFM.2006.27}.

\bibitemdeclare{inproceedings}{BliSif07-acp-emsoft}
\bibitem{BliSif07-acp-emsoft}
\bibinfo{author}{Simon \surnamestart Bliudze\surnameend} \&
  \bibinfo{author}{Joseph \surnamestart Sifakis\surnameend}
  (\bibinfo{year}{2007}): \emph{\bibinfo{title}{The algebra of connectors:
  structuring interaction in {BIP}}}.
\newblock In: {\sl \bibinfo{booktitle}{Proc. of EMSOFT}},
  \bibinfo{organization}{ACM SigBED}, \bibinfo{publisher}{ACM},
  \bibinfo{address}{Salzburg, Austria}, pp. \bibinfo{pages}{11--20},
  \doi{10.1145/1289927.1289935}.

\bibitemdeclare{inproceedings}{BBJS14}
\bibitem{BBJS14}
\bibinfo{author}{Simon \surnamestart Bliudze\surnameend},
  \bibinfo{author}{Joseph \surnamestart Sifakis\surnameend},
  \bibinfo{author}{Marius \surnamestart Bozga\surnameend} \&
  \bibinfo{author}{Mohamad \surnamestart Jaber\surnameend}
  (\bibinfo{year}{2014}): \emph{\bibinfo{title}{Architecture Internalisation in
  {BIP}}}.
\newblock In: {\sl \bibinfo{booktitle}{Proc. of CBSE}},
  \bibinfo{publisher}{ACM}, pp. \bibinfo{pages}{169--178},
  \doi{10.1145/2602458.2602477}.

\bibitemdeclare{inproceedings}{BMM11}
\bibitem{BMM11}
\bibinfo{author}{Roberto \surnamestart Bruni\surnameend},
  \bibinfo{author}{Hern\'an \surnamestart Melgratti\surnameend} \&
  \bibinfo{author}{Ugo \surnamestart Montanari\surnameend}
  (\bibinfo{year}{2011}): \emph{\bibinfo{title}{Connector Algebras, {P}etri
  {N}ets, and {BIP}}}.
\newblock In: {\sl \bibinfo{booktitle}{Proc. of PSI}}, {\sl
  \bibinfo{series}{LNCS}} \bibinfo{volume}{7162},
  \bibinfo{publisher}{Springer}, pp. \bibinfo{pages}{19--38},
  \doi{10.1007/978-3-642-29709-0\_2}.

\bibitemdeclare{incollection}{CRB+08}
\bibitem{CRB+08}
\bibinfo{author}{M.~Y. \surnamestart Chkouri\surnameend},
  \bibinfo{author}{A.~\surnamestart Robert\surnameend},
  \bibinfo{author}{M.~\surnamestart Bozga\surnameend} \&
  \bibinfo{author}{J.~\surnamestart Sifakis\surnameend} (\bibinfo{year}{2009}):
  \emph{\bibinfo{title}{Translating {AADL} into {BIP} - Application to the
  Verification of Real-Time Systems}}.
\newblock In: {\sl \bibinfo{booktitle}{Proc. of MODELS}}, {\sl
  \bibinfo{series}{LNCS}} \bibinfo{volume}{5421},
  \bibinfo{publisher}{Springer}, pp. \bibinfo{pages}{5--19},
  \doi{10.1007/978-3-642-01648-6\_2}.

\bibitemdeclare{techreport}{bip2reo}
\bibitem{bip2reo}
\bibinfo{author}{K.~\surnamestart Dokter\surnameend},
  \bibinfo{author}{S.-S.~T.Q. \surnamestart Jongmans\surnameend},
  \bibinfo{author}{F.~\surnamestart Arbab\surnameend} \&
  \bibinfo{author}{S.~\surnamestart Bliudze\surnameend} (\bibinfo{year}{2015}):
  \emph{\bibinfo{title}{Relating {BIP} and {Reo}}}.
\newblock \bibinfo{type}{Technical Report} \bibinfo{number}{FM-1505},
  \bibinfo{institution}{CWI}.
\newblock
  \urlprefix\url{http://persistent-identifier.org/?identifier=urn:nbn:nl:ui:18%
-23505}.

\bibitemdeclare{inproceedings}{Garlan}
\bibitem{Garlan}
\bibinfo{author}{David \surnamestart Garlan\surnameend} (\bibinfo{year}{2014}):
  \emph{\bibinfo{title}{Software Architecture: A Travelogue}}.
\newblock In: {\sl \bibinfo{booktitle}{Proc. of FOSE}},
  \bibinfo{publisher}{ACM}, pp. \bibinfo{pages}{29--39},
  \doi{10.1145/2593882.2593886}.

\bibitemdeclare{article}{JA12}
\bibitem{JA12}
\bibinfo{author}{Sung-Shik T.~Q. \surnamestart Jongmans\surnameend} \&
  \bibinfo{author}{Farhad \surnamestart Arbab\surnameend}
  (\bibinfo{year}{2012}): \emph{\bibinfo{title}{Overview of Thirty Semantic
  Formalisms for {R}eo}}.
\newblock {\sl \bibinfo{journal}{Sci. Ann. Comp. Sci.}}
  \bibinfo{volume}{22}(\bibinfo{number}{1}), pp. \bibinfo{pages}{201--251},
  \doi{10.7561/SACS.2012.1.201}.

\bibitemdeclare{inproceedings}{KC09}
\bibitem{KC09}
\bibinfo{author}{C.~\surnamestart Koehler\surnameend} \&
  \bibinfo{author}{D.~\surnamestart Clarke\surnameend} (\bibinfo{year}{2009}):
  \emph{\bibinfo{title}{Decomposing port automata}}.
\newblock In: {\sl \bibinfo{booktitle}{Proc. of SAC}},
  \bibinfo{publisher}{ACM}, pp. \bibinfo{pages}{1369--1373},
  \doi{10.1145/1529282.1529587}.

\bibitemdeclare{inproceedings}{Krause09}
\bibitem{Krause09}
\bibinfo{author}{C.~\surnamestart Krause\surnameend} (\bibinfo{year}{2009}):
  \emph{\bibinfo{title}{Integrated Structure and Semantics for {R}eo Connectors
  and Petri Nets}}.
\newblock In: {\sl \bibinfo{booktitle}{Proc. of ICE}}, pp.
  \bibinfo{pages}{57--69}, \doi{10.4204/EPTCS.12.4}.

\bibitemdeclare{book}{Milner89}
\bibitem{Milner89}
\bibinfo{author}{R.~\surnamestart Milner\surnameend} (\bibinfo{year}{1989}):
  \emph{\bibinfo{title}{Communication and Concurrency}}.
\newblock \bibinfo{publisher}{Prentice-Hall, Inc.}

\bibitemdeclare{article}{PA01}
\bibitem{PA01}
\bibinfo{author}{G.~A. \surnamestart Papadopoulos\surnameend} \&
  \bibinfo{author}{F.~\surnamestart Arbab\surnameend} (\bibinfo{year}{2001}):
  \emph{\bibinfo{title}{Configuration And Dynamic Reconfiguration Of Components
  Using The Coordination Paradigm}}.
\newblock {\sl \bibinfo{journal}{Future Generation Computer Systems}}
  \bibinfo{volume}{17}(\bibinfo{number}{8}), pp. \bibinfo{pages}{1023 -- 1038},
  \doi{10.1016/S0167-739X(01)00043-7}.

\bibitemdeclare{article}{PC08}
\bibitem{PC08}
\bibinfo{author}{Jos{\'e} \surnamestart Proen\c{c}a\surnameend} \&
  \bibinfo{author}{Dave \surnamestart Clarke\surnameend}
  (\bibinfo{year}{2008}): \emph{\bibinfo{title}{Coordination Models {O}rc and
  {R}eo Compared}}.
\newblock {\sl \bibinfo{journal}{Electron. Notes Theor. Comput. Sci.}}
  \bibinfo{volume}{194}(\bibinfo{number}{4}), pp. \bibinfo{pages}{57--76},
  \doi{10.1016/j.entcs.2008.03.099}.

\bibitemdeclare{article}{TSR11}
\bibitem{TSR11}
\bibinfo{author}{Carolyn \surnamestart Talcott\surnameend},
  \bibinfo{author}{Marjan \surnamestart Sirjani\surnameend} \&
  \bibinfo{author}{Shangping \surnamestart Ren\surnameend}
  (\bibinfo{year}{2011}): \emph{\bibinfo{title}{Comparing three coordination
  models: {R}eo, {ARC}, and {PBRD}}}.
\newblock {\sl \bibinfo{journal}{Sci. Comput. Programming}}
  \bibinfo{volume}{76}(\bibinfo{number}{1}), pp. \bibinfo{pages}{3--22},
  \doi{10.1016/j.scico.2009.11.006}.

\end{thebibliography}
\bibliographystyle{eptcs}

\end{document}